\documentclass{llncs}

\usepackage[fixed]{fontawesome6}

\usepackage{graphicx}

\usepackage{tikz}
\usetikzlibrary{tikzmark}
\usetikzlibrary{calc}
\usetikzlibrary{arrows.meta}
\usetikzlibrary{shapes.geometric}
\usetikzlibrary{patterns}
\usetikzlibrary{decorations.pathmorphing}
\usetikzlibrary{angles}

\usepackage{amsmath,amssymb}
\usepackage{stmaryrd}
\usepackage{mathtools}

\usepackage{thmtools,thm-restate}

\usepackage{booktabs}
\usepackage{nicematrix}
\usepackage{nicefrac}
\usepackage{enumitem}

\usepackage{algorithm}
\usepackage[noend]{algpseudocode}
\usepackage{hyperref}

\usepackage{cleveref}

\usepackage{listings}
\title{Implicit Computation of Filtered Prime Implicants}
\author{
    Edward Liem \and %
    Clemens Dubslaff
}
\institute{
    Eindhoven University of Technology, The Netherlands\\
    \email{\{e.liem,c.dubslaff\}@tue.nl}
}

\begin{document}
\maketitle
\begin{abstract}
Prime implicants (PIs) are central in computer science, with applications in logic minimization, diagnosis, explainable formal methods and AI.
Algorithms for the computation of PIs were first-and-foremost considered on the full input space, not considering the case where the input space might be constrained by context or structural dependencies.
To filter out those PIs that do not fulfill the constraints, existing methods rely on an explicit post-processing step after computing all PIs, which leads to scalability issues due to the number of PIs being exponential.
We propose end-to-end symbolic algorithms that implicitly represent the set of PIs under side constraints.
For this, we extend the prominent method based on decision diagrams by Coudert and Madre and implement a modular tool chain that separates decision-diagram construction, PI computation, and filtering.
\end{abstract}

\section{Introduction}\label{sec:intro}
A prime implicant (PI) of a Boolean function is a minimal conjunction of literals that implies the function.
Prime implicants play an important role in many areas of computer science.
For example, to reduce the size of circuits, one is usually interested in a small collection of terms that suffice to realize an output.
In fault diagnosis, a maximal set of components allowed to fail is important to assess the robustness of the system.
For configurable software systems, we may want to know which configurations yield an effect such as high energy consumption.
All these applications mainly rely on computing sets of PIs.
Minimal sum-of-products for circuits via PIs have been prominently considered by Quine and McCluskey~\cite{Quine1952,McCluskey1956}, followed by further applications to limit the search space in two-level logic minimization~\cite{Brayton1984,Coudert1994}.
In model-based diagnosis, PIs coincide with minimal conflicts and minimal supports~\cite{KleRei87,KleWil87,Coudert1993}.
In the emerging field of explainable formal methods and artificial intelligence (AI) they are the mathematical definition of an \emph{abductive explanation}~\cite{Shih2018,Ignatiev2019} or \emph{feature cause}~\cite{Dubslaff2022,Dubslaff2024}: a minimal set of feature values that entail an effect, independently of all remaining features.

However, computing PIs is challenging, as a Boolean function may possess exponentially many PIs~\cite{ChaMar78}.
Deciding whether a shorter implicant exists given a PI candidate is $\Sigma_2^p$-complete~\cite{Umans2001} and the naive approach enumerating all implicants and then testing each for minimality is only feasible for small examples or for very specific subclasses of Boolean functions.
Practical algorithms therefore rely on \emph{symbolic methods} to represent sets of PIs \emph{implicitly} as a decision diagram~\cite{Coudert1992,Minato1993}.
Orthogonal \emph{search-based} methods avoid the construction of the whole set of PIs and compute individual primes only on demand, shrinking candidates through a sequence of calls to a satisfiability oracle~\cite{Manquinho1997,Deharbe2013,Previti2015}.

Standard approaches target the computation of PIs for an \emph{unconstrained} Boolean function, ranging over the full set of possible variable assignments as inputs while potentially allowing for \emph{don't cares}.
However, many applications such as covering problems rely on solving constrained sub problems.
For example, offset-based methods for two-level logic minimization iteratively compute PIs that cover a minterm, constraining the PIs that only partly cover the minterm~\cite{MalBraNew91}.
Also within PI-based explanations, constraints naturally arise when the set of inputs are restricted to \emph{valid} ones, ruling out all PIs that only cover invalid inputs or inputs a user does not care about~\cite{Shih2018,Cooper2023,Dubslaff2022}.
For these applications, existing computation approaches \emph{explicitly} iterate over the set of all PIs, filtering those that do not fulfill the constraints as a post-processing step.
Therefore, current implementations scaled only to small instances of filtered PI-computations~\cite{Dubslaff2022,Dubslaff2024,HasDubWil26}.

In this paper, we tackle the scalability problem by establishing an \emph{end-to-end symbolic algorithm} to compute filtered PIs, realized through symbolic filters on the implicit representation of all PIs.
For this, we rely on the prominent decision-diagram algorithm to compute PIs by Coudert and Madre~\cite{Coudert1992} and extend it by carefully operating on the implicit PI-set representation.
We present three variants of filters to solve constrained PI problems: an \emph{existential filter} that only includes PIs that touch an input we care about, useful for computing feature causes~\cite{Dubslaff2024}, a \emph{universal filter} that requires the PI only covering a care set of inputs, which has natural applications for constrained abductive explanations~\cite{Cooper2023}, and a \emph{subset filter} that is useful to compute offset PIs and most general feature causes~\cite{MalBraNew91,Dubslaff2022}.
We implement our approach in a modular tool chain as depicted in \Cref{fig:pipeline}, separating decision-diagram construction, PI-computation, and our novel filtering methods.

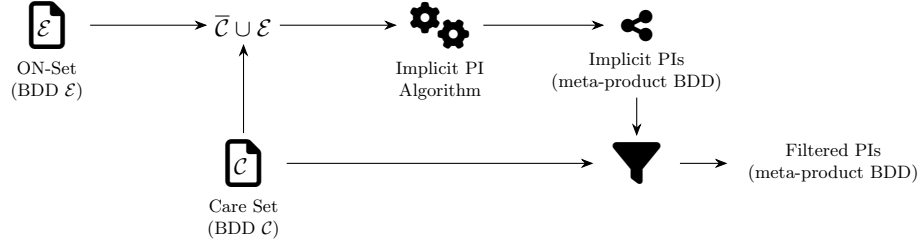
\begin{figure}[t]
{
\noindent
\begin{center}
\scalebox{1}{
\begin{tikzpicture}
    \matrix[
        column sep={8em, between origins},
        row sep=30pt,
        every node/.style={anchor=center,inner sep=1pt},
        ampersand replacement=\&,
    ] {
        \node[scale=2] (valid) {
            \faFile[regular]
        }; 
        \& 
        \node[outer sep=2pt,scale=1] (eq) {
            $\overline{\mathcal{C}}\cup\mathcal{E}$
        };
        \&
        \node[
            scale=2,
        ] (algo) {
            \faGears
        };
        \&
        \node[
            scale=1.5
        ] (bdd) {
            \faShareNodes
        };
\&
        \\

        \&
        \node[scale=2] (effect) { 
            \faFile[regular]
        };
        \&
        \&
        \node[scale=2] (filter) {
            \faFilter 
        };
        \&
        \node[scale=0.75,outer sep=1em,align=center] (featCause){
            Filtered PIs \\
            (meta-product BDD)
        };\\
    };

    \node [scale=0.75,align=center,anchor=north] (bdd_text) at (bdd.south) {
        Implicit PIs\\
        (meta-product BDD)
    };
    \node[anchor = center, scale=.95, yshift=-1pt, xshift=-1pt] at (valid.center) {$\mathcal{E}$};
    \node[anchor = north, scale=.75, align=center] (valid text) at (valid.south) {
        ON-Set\\ (BDD $\mathcal{E}$)
    };
    \node[anchor = center, scale=.95,yshift=-1pt, xshift=-1pt] at (effect.center) {$\mathcal{C}$};
    \node[anchor = north, scale=.75, align=center] (effect text) at (effect.south) {
        Care Set\\ (BDD $\mathcal{C}$)
    };
    \node [scale=0.75,align=center,anchor=north] at (algo.south) {
        Implicit PI\\
        Algorithm
    };

    \draw[-Stealth] (valid) -- (eq);
    \draw[-Stealth] (effect) -- (eq);
    \draw[-Stealth] (effect) -- (filter);
    \draw[-Stealth] (eq) -- (algo);
    \draw[-Stealth] (algo) -- (bdd);
    \draw[-Stealth] (bdd_text) -- (filter);
    \draw[-Stealth] (filter) -- (featCause);
\end{tikzpicture}
}
\end{center}
}%
\caption{Schema of the pipeline of the implicit existential filter on prime implicants.\label{fig:pipeline}}
\end{figure}

\section{Preliminaries}\label{sec:prelim}

\subsection{Assignments}
\def\pAs{\ensuremath\Delta_X}
\def\tAs{\ensuremath\Theta_X}
Let \(X = \{x_{1},\dots,x_{n}\}\) be a finite set of \(n\) variables.
A (partial) \emph{assignment} of \(X\) is a partial mapping \(\mathbf{p}\colon X \rightharpoonup \mathbb{B}\) where \(\mathbb{B} = \{0,1\}\) is a Boolean domain.
We write \(\mathbf{p}(x) = *\) if \(\mathbf{p}\) is not defined for \(x \in X\).
The \emph{support} of \(\mathbf{p}\), notated \(\mathsf{supp}(\mathbf{p})\), is the set of elements \(x \in X\) where \(\mathbf{p}(x) \neq *\).
An assignment is %
\emph{total} if \(\mathsf{supp}(\mathbf{p}) = X\).
The set of all partial and total assignments over \(X\) is defined as \(\pAs\) and \(\tAs\), respectively.
We write \(\mathbf{p}[v/x]\) to denote the setting of variable \(x\) to \(v \in \mathbb{B} \cup \{*\}\) in assignment \(\mathbf{p}\), i.e.\ \(\mathbf{p}[v/x](x) = v\) and \(\mathbf{p}[v/x](y) = \mathbf{p}(y)\) for \(x,y \in X\) and \(x \neq y\).
As a shorthand, we may write an assignment as a sequence of variables where each variable appears at most once and may appear positively (\(x\)) or negatively (\(\overline{x}\)); for instance, \(\mathbf{p} = x\overline{y}\) is the assignment for which \(\mathbf{p}(x) = 1\), \(\mathbf{p}(y) = 0\) and \(\mathbf{p}(z) = *\) for all \(z \in X\setminus\{x,y\}\).

Given an assignment \(\mathbf{p} \in \pAs\), we define its \emph{cover} to be \(\llbracket{\mathbf{p}}\rrbracket= \{\mathbf{u} \in \tAs \mid \forall x \in \mathsf{supp}(\mathbf{p}).\ \mathbf{p}(x) = \mathbf{u}(x)\}\).
For a set of partial assignments \(P \subseteq \pAs\), we define its cover to be \(\llbracket{P}\rrbracket = \bigcup_{\mathbf{p} \in P}\llbracket{\mathbf{p}}\rrbracket\).

\subsection{Boolean Functions}
\def\boolTrue{\ensuremath{\mathtt{true}}}
\def\boolFalse{\ensuremath{\mathtt{false}}}
A \emph{Boolean function} \(f\colon \tAs \rightarrow \mathbb{B}\) maps total assignments to \(0\) or \(1\).
The \emph{cover} of a function is defined as \(\llbracket{f}\rrbracket = \{\mathbf{u} \in \tAs \mid f(\mathbf{u}) = 1\}\).
The semantic equivalence of two Boolean functions \(f,g\), written as \(f \equiv g\), is defined such \(f(\mathbf{u}) = g(\mathbf{u})\) for all \(\mathbf{u} \in \tAs\).
We define two special Boolean functions \boolTrue{} and \boolFalse{}, called \emph{constants}, where for all \(\mathbf{u} \in \tAs\) %
we have \(\boolTrue(\mathbf{u}) = 1 \text{ and } \boolFalse(\mathbf{u}) = 0\), respectively.
We write \(g = x\) if \(g(\mathbf{u}[1/x]) = 1\) for all \(\mathbf{u} \in \tAs\) and denote by \(f \land g\) the conjunction, \(f \lor g\) the disjunction, and \(\overline{f}\) the negation of Boolean functions with the usual semantics.
The \emph{restriction} of a variable \(x \in X\) to \(b \in \mathbb{B}\) in \(f\) is defined as \(f[b/x](\mathbf{u}) = f(\mathbf{u}[b/x])\) for all \(\mathbf{u} \in \tAs\).
The \emph{negative} and \emph{positive cofactors} of \(f\) w.r.t.\ \(x\) are defined as \(f_{\overline{x}} = f[0/x]\) and \(f_{x} = f[1/x]\), respectively.
The \emph{Shannon decomposition} of \(f\) w.r.t.\ \(x\) is the following identity~\cite{Shannon1949}:
\[
f \equiv (x \land f_{x}) \lor (\overline{x} \land f_{\overline{x}})
\]
The \emph{characteristic function} of \(S \subseteq \tAs\) is a unique Boolean function \(c_{S}\) such that \(c_{S}(\mathbf{u}) = 1\) iff \(\mathbf{u} \in S\).
A partial assignment \(\mathbf{p} \in \pAs\) is an \emph{implicant} of Boolean function \(f\) iff \(f(\mathbf{u}) = 1\) for all \(\mathbf{u} \in \llbracket{\mathbf{p}}\rrbracket\).
An implicant \(\mathbf{p} \in \pAs\) of \(f\) is \emph{prime} iff \(\llbracket{\mathbf{p}}\rrbracket \not\subset \llbracket{\mathbf{q}}\rrbracket\) for all implicants \(\mathbf{q}\) of \(f\).

\subsection{Meta-products}
\def\tAsO{\ensuremath{\Theta_{O}}}
\def\tAsS{\ensuremath{\Theta_{S}}}
A Boolean function over variables \(X\) is unable to represent sets of partial assignments over \(X\).
Coudert and Madre define the notion of a \emph{meta-product} (MP) to encode a set of partial assignments as Boolean functions~\cite{Coudert1992}.
Observe that for an assignment \(\mathbf{p} \in \pAs\), an element \(x \in X\) may be \(\mathbf{p}(x) = 1\), \(\mathbf{p}(x) = 0\), or \(\mathbf{p}(x) = *\).
As the set of assignments \(\pAs\) has \(3^{n}\) elements, \(\log_{2}3^{n}\) Boolean variables suffice to represent any product in \(\pAs\).
However, Coudert and Madre has opted to encode assignment sets with \(2n\) additional variables to maintain direct links between assignment sets and its representation.
In particular, for every variable \(x_{i}\), one additional Boolean variable is used to indicate whether it appears in the assignment, and one additional variable is used to indicate whether the variable \(x_{i}\) is assigned to \(0\) or \(1\).

Given a variable set \(X = \{x_{1}, \dots, x_{n}\}\), define \(2n\) additional variables for \emph{occurrence} \(O = \{o_{1}, \dots, o_{n}\}\) and \emph{sign} \(S = \{s_{1}, \dots, s_{n}\}\) such that \(X\), \(O\), and \(S\) are all mutually exclusive.
Since variables \(o_{i}\) and \(s_{i}\) are linked to \(x_{i}\), let \(\mathsf{var}(o_{i}) = \mathsf{var}(s_{i}) = x_{i}\).
Moreover, we abuse notation and say \(\mathbf{o}(x_{i}) = \mathbf{o}(o_{i})\) and \(\mathbf{s}(x_{i}) = \mathbf{s}(s_{i})\) for \(\mathbf{o} \in \tAsO\), \(\mathbf{s} \in \tAsS\).
Let \(\sigma\colon \tAsO \times \tAsS \rightarrow \pAs\) be a surjective mapping from pairs of occurrence and sign variables to products such that \(\sigma(\mathbf{o},\mathbf{s}) = \mathbf{p}\) iff for all \(x \in X\), 
{
\setlength{\abovedisplayskip}{5pt}
\setlength{\belowdisplayskip}{5pt}
\begin{align*}
&(\mathbf{p}(x) = * \Longleftrightarrow \mathbf{o}(x) = 0) \land {}\\
&(\mathbf{p}(x) = 0 \Longleftrightarrow (\mathbf{o}(x) = 1 \land \mathbf{s}(x) = 0)) \land {}\\
&(\mathbf{p}(x) = 1 \Longleftrightarrow (\mathbf{o}(x) = 1 \land \mathbf{s}(x) = 1))
\end{align*}
}%
In other words, we have that whenever \(\mathbf{p}(x) = *\), it must be that \(x\) does not occur, i.e.\ \(\mathbf{o}(x) = 0\).
If \(\mathbf{p}(x) = b\) for \(b \in \mathbb{B}\), it has to be that \(x\) does occur (\(\mathbf{o}(x) = 1\)) and its setting must be \(b\), i.e.\ \(\mathbf{s}(x) = b\).
Observe that \(\sigma\) may map multiple \(\mathbf{o},\mathbf{s}\) instance pairs to the same product \(\mathbf{p}\).
\begin{example}
Let \(\mathbf{p}\) be a partial assignment over \(X = \{x_{1},x_{2},x_{3}\}\) such that \(\mathbf{p} = x_{1}\overline{x_{3}}\).
We have that \(\sigma(o_{1}\overline{o_{2}}o_{3}, s_{1}s_{2}\overline{s_{3}}) = \sigma(o_{1}\overline{o_{2}}o_{3}, s_{1}\overline{s_{2}}\overline{s_{3}}) = \mathbf{p}\) since the sign variable of \(x_{2}\) plays no role.
\end{example}

Let \(\sigma^{-1}\colon \pAs \rightarrow 2^{(\tAsO \times \tAsS)}\) be the inverse function of \(\sigma\) such that \(\sigma^{-1}(\mathbf{p}) = \{(\mathbf{o,s}) \mid  \sigma(\mathbf{o,s}) = \mathbf{p}\}\) for all \(\mathbf{p} \in \tAs\).
The Boolean function \(m_{P} \colon (\tAsO \times \tAsS) \rightarrow \mathbb{B}\) is the \emph{meta-product} of \(P \subseteq \pAs\) such that \(m_{P}\) is the characteristic function of 
\[
    \Bigg( \bigcup_{\mathbf{p} \in P} \sigma^{-1}(\mathbf{p}) \Bigg)
\]
Meta-products are a canonical functional representation for the set of assignments \(P \subseteq \pAs\)~\cite{Madre1992}.

\def\bddFalse{\ensuremath\bot}
\def\bddTrue{\ensuremath\top}
\subsection{Binary Decision Diagrams}
\emph{Binary decision diagrams} (BDDs) offer a representation for Boolean functions.
A BDDs is an acyclic, directed graph \(\mathcal{G} = (T,V,\lambda,\delta,\rho)\) such that
\begin{itemize}
    \item \(\varnothing \subset T \subseteq \{\textbf{0},\textbf{1}\}\) is a set of \emph{terminals}, 
    \item \(V\) is a finite set of \emph{vertices}, where \(V \cap T \neq \varnothing\), 
    \item \(\lambda\colon V \rightarrow X\) is a \emph{labeling function}, 
    \item \(\delta\colon (V \times \mathbb{B}) \rightarrow (V \cup T)\) is a \emph{decision function}, and
    \item a \emph{root} \(\rho \in (V \cup T)\).
\end{itemize}
The semantics \(\llbracket{u}\rrbracket\) of some \(u \in (V \cup T)\) of BDD \(\mathcal{G} = (T,V,\lambda,\delta,\rho)\) is recursively defined as follows:
{
\setlength{\abovedisplayskip}{5pt}
\setlength{\belowdisplayskip}{5pt}
\begin{align*}
\llbracket{\textbf{0}}\rrbracket &\equiv \boolFalse \\
\llbracket{\textbf{1}}\rrbracket &\equiv \boolTrue \\
\llbracket{v}\rrbracket &\equiv (\overline{\lambda(v)} \land \llbracket{\delta(v,0)}\rrbracket) \lor (\lambda(v) \land \llbracket{\delta(v,1)}\rrbracket)\text{ for } v \in V
\end{align*}
}%
BDD \(\mathcal{G} = (T,V,\lambda,\delta,\rho)\) is semantically equivalent to Boolean function \(g\), notated as \(\llbracket{\mathcal{G}}\rrbracket \equiv g\), iff \(\llbracket{\rho}\rrbracket \equiv g\).
As a shorthand, \(\bddFalse\) is a BDD such that \(\llbracket{\bddFalse}\rrbracket \equiv \boolFalse\), and \(\bddTrue\) is a BDD where \(\llbracket{\bddTrue}\rrbracket \equiv \boolTrue\).
A BDD \(\mathcal{G}\) is \emph{ordered} if for any vertex \(v \in V\) such that \(\lambda(v) = x_{i}\) and vertex \(u \in \{\delta(v,0), \delta(v,1)\}\) such that \(\lambda(u) = x_{j}\), we must have that \(i < j\).
Moreover, an ordered BDD \(\mathcal{G}\) is \emph{reduced} if there does not exist two distinct isomorphic subgraphs of \(\mathcal{G}\) and for all vertices \(v \in V\), \(\delta(v,0) \neq \delta(v,1)\).
Henceforth, we assume all BDDs are reduced and ordered.
By fixing a variable ordering, BDDs offer a canonical representation of Boolean functions~\cite{Bryant1986}. 
Thus, we assume a fixed ordering over variables, and for every Boolean function \(f\) we associate the unique BDD \(\mathcal{F}\) such that \(\llbracket{\mathcal{F}}\rrbracket \equiv f\).
We shall notate Boolean functions using lower-case letters and their semantically equivalent BDD by corresponding script upper-case letters.
Additionally, for non-constant Boolean functions \(f\) and BDD \(\mathcal{F} = (T,V,\lambda,\delta,\rho)\) such that \(\llbracket{\mathcal{F}}\rrbracket \equiv f\) and $\lambda(\rho)=x$, it follows by Shannon decomposition that \(f_{\overline{x}} \equiv \llbracket{\delta(\rho,0)}\rrbracket\) and \(f_{x} \equiv \llbracket{\delta(\rho,1)}\rrbracket\); as a shorthand, we use \(\mathcal{F}_{\overline{x}}\) and \(\mathcal{F}_{x}\) to represent sub-BDDs \(\delta(\rho,0)\) and \(\delta(\rho,1)\) respectively.
Furthermore, we say that \(\mathcal{F}_{x}\) and \(\mathcal{F}_{\overline{x}}\) are the positive and negative cofactors of BDD \(\mathcal{F}\), respectively.
We assume the standard Boolean connectives and operators applies for BDDs.
The restriction of variable \(x \in X\) to \(b \in \mathbb{B}\) in BDD \(\mathcal{F}\), written \(\mathcal{F}[b/x]\), is defined such that \(\llbracket{\mathcal{F}[b/x]}\rrbracket \equiv f[b/x]\) for \(\llbracket{\mathcal{F}}\rrbracket \equiv f\).

\paragraph{BDDs for meta-products.}
As BDD sizes are sensitive to the ordering of variables~\cite{Bryant1986}, we fix the following ordering of meta-product variables to minimize BDD representation sizes and the computational cost of BDD operators over meta-products:
\[
    o_{1} < s_{1} < x_{1} < \dots < o_{n} < s_{n} < x_{n}
\]
We associate every set \(P \subseteq \pAs\) to a meta-product BDD \(\mathcal{P}\) that is uniquely defined by \(\llbracket{\mathcal{P}}\rrbracket \equiv m_{P}\).

\subsection{Implicit Prime Implicants}
Coudert and Madre presented implicit techniques to manipulate primes and derived the following theorem about meta-products:
\begin{theorem}[\cite{Coudert1992}]\label{thm:CM_primes}
The meta-product \(\mathsf{Prime}(f)\) of the set of prime implicants of the function \(f \equiv (\overline{x} \land f_{\overline{x}}) \lor (x \land f_{x})\) is 
\begin{align*}
&(\overline{o_{x}} \land \mathsf{Prime}(f_{\overline{x}} \land f_{x})) \lor {}\\
&(o_{x} \land \overline{s_{x}} \land \mathsf{Prime}(f_{\overline{x}}) \land \overline{\mathsf{Prime}(f_{\overline{x}} \land f_{x})}) \lor {}\\
&(o_{x} \land s_{x} \land \mathsf{Prime}(f_{x}) \land \overline{\mathsf{Prime}(f_{\overline{x}} \land f_{x})})
\end{align*}
\end{theorem}
Using \Cref{thm:CM_primes}, an algorithm over BDDs may naturally be derived.
Given a BDD \({\mathcal{F}}\), \Cref{algo:cmPrime} outputs the meta-product BDD \(\mathcal{P}\) such that \(\llbracket{\mathcal{P}}\rrbracket \equiv m_{P}\) for \(P \subseteq \pAs\) and for all \(\mathbf{p} \in P\), \(\mathbf{p}\) is a prime implicant of \(f\), where \(\llbracket{\mathcal{F}}\rrbracket \equiv f\).
This algorithm computes prime implicants for \(f\) by recursively considering the primes of the positive cofactor (\(f_{x}\)), negative cofactor (\(f_{\overline{x}}\)), and the case where \(x\) does not affect the output of \(f\) (\(f_{x} \land f_{\overline{x}}\)).

\newcommand{\Prime}{\textsc{Prime}}
\begin{algorithm}
	\caption{$\Prime(\mathcal{F})$}\label{algo:cmPrime}
	\textbf{Input: } BDD $\mathcal{F} = (T, V, \lambda, \delta, \rho)$ \\
	\textbf{Output: } meta-product BDD $\mathcal{P}$
	\begin{algorithmic}[1]
		\If{$\llbracket{\mathcal{F}}\rrbracket \equiv \boolTrue$} {\Return{$\overline{o_1} \land \dots \land \overline{o_n}$}}
		\ElsIf{$\llbracket{\mathcal{F}}\rrbracket \equiv \boolFalse$} \Return{\bddFalse}
		\Else%
		\State{$x \leftarrow \lambda(\rho)$}\vspace{1em}
		\State{$\mathcal{P}_{*} \leftarrow \Prime(\mathcal{F}_x \land \mathcal{F}_{\overline{x}})$}
		\State{$\mathcal{P}_{0} \leftarrow \Prime(\mathcal{F}_{\overline{x}})$}
		\State{$\mathcal{P}_{1} \leftarrow \Prime(\mathcal{F}_x)$}\vspace{1em}

		\State{\Return{
				\begin{minipage}[t]{0.5\textwidth}
					$(\overline{o_x} \land \mathcal{P}_{*}[0/o_x]) \lor {}$ \\
					$(o_x \land \overline{s_x} \land \mathcal{P}_{0}[0/o_{x}] \land \overline{\mathcal{P}_{*}[0/o_x]}) \lor {}$\\
					$(o_x \land s_x \land \mathcal{P}_1[0/o_{x}] \land \overline{\mathcal{P}_{*}[0/o_x]})$
				\end{minipage}
			}}
		\EndIf%
	\end{algorithmic}
\end{algorithm}

\section{Filtering Assignment Sets}\label{sec:filter}
In this section, we present algorithms that operate on the structure of meta-product BDDs to filter and remove partial assignments not meeting certain criteria.
Given a set of assignments \(P \subseteq \pAs\) and a Boolean function \(f\), we would like to obtain the set \(Q \subseteq P\) such that for all \(\mathbf{q} \in Q\), the predicate \(\Phi(\mathbf{q},f)\) holds.
In particular, we define three filtering predicates \(\Phi=\{\exists,\forall,\supseteq\}\):
\begin{align*}
\exists(\mathbf{q},f) &= \exists \mathbf{u} \in \llbracket{f}\rrbracket.\ \mathbf{u} \in \llbracket{\mathbf{q}}\rrbracket\\
\forall(\mathbf{q},f) &= \forall \mathbf{u} \in \llbracket{f}\rrbracket.\ \mathbf{u} \in \llbracket{\mathbf{q}}\rrbracket\\
{\supseteq}(\mathbf{q},f) &= \forall \mathbf{u} \in \llbracket{\mathbf{q}}\rrbracket.\ \mathbf{u} \in \llbracket{f}\rrbracket
\end{align*}

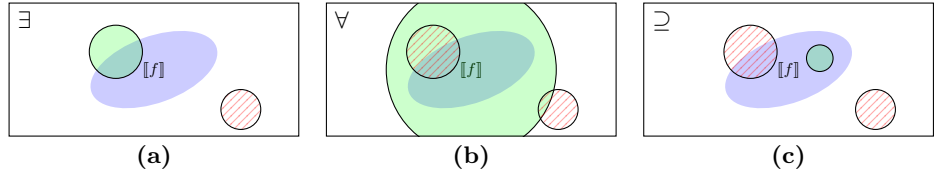
\begin{figure}[t]
\centering
{
\newlength{\picGap} \setlength{\picGap}{10pt}
\newlength{\boxWd} \setlength{\boxWd}{\dimexpr((\textwidth-(\picGap*2))/3)\relax}
\newlength{\boxHt} \setlength{\boxHt}{50pt}
\begin{tikzpicture}[
    picBox/.style={
        minimum width=\boxWd,
        minimum height=\boxHt,
        draw,
        outer sep=0pt,
        inner sep=0pt,
    }
]
\matrix[
    column sep=\picGap,
    opacity=0,
] {
    \node[picBox] (eFrame) {}; & 
    \node[picBox] (aFrame) {}; & 
    \node[picBox] (cFrame) {}; \\
};

\newcounter{subfigCounter} \setcounter{subfigCounter}{1}
\foreach \x / \y in {eFrame/\exists,aFrame/\forall,cFrame/\supseteq} {
    \begin{scope}[
        valid/.style={
            draw,
            circle,
            fill=green,
            fill opacity=0.2,
        },
        invalid/.style={
            draw,
            circle,
            fill opacity=0.33,
            pattern=north east lines,
            pattern color=red,
        },
    ]
        \clip (\x.north west) -- (\x.north east) -- (\x.south east) -- (\x.south west);
        \node[
            shape=ellipse,
            inner sep=0pt,
            outer sep=0pt,
            minimum width=50pt,
            minimum height=25pt,
            fill=blue,
            fill opacity=0.2,
            rotate=20,
        ] (f) at (\x) {};
        \node[
            scale=0.66
        ] at (f) {\(\llbracket{f}\rrbracket\)};
    
        \ifnum\pdfstrcmp{\x}{eFrame}=0%
            \node[
                valid,
                minimum width=20pt
            ] at (f.135) {};
            \node[invalid,minimum width=15pt] at ($(\x)!0.6!(\x.south east)$) {};
        \fi
        \ifnum\pdfstrcmp{\x}{aFrame}=0%
            \node[
                valid,
                minimum width=64pt,
            ] at (\x) {};
            \node[invalid, minimum width=20pt] at (f.135) {};
            \node[invalid,minimum width=15pt] at ($(\x)!0.6!(\x.south east)$) {};
        \fi
        \ifnum\pdfstrcmp{\x}{cFrame}=0%
            \node[
                valid,
                minimum width=10pt
            ] at ($(f)!0.5!(f.0)$) {};
            \node[invalid, minimum width=20pt] at (f.135) {};
            \node[invalid,minimum width=15pt] at ($(\x)!0.6!(\x.south east)$) {};
        \fi

        \node[
            anchor=north west,
        ] at (\x.north west) {\(\y\)};
    \end{scope}
    \node[picBox] (box) at (\x) {};
    \node[anchor=north] at (box.south) {
        \textbf{(\alph{subfigCounter})}
    };

    \stepcounter{subfigCounter}
}

\end{tikzpicture}
}
\vskip-15pt
\caption{Illustration of our filtering operations in an assignment space.
The covers of the function $f$ are depicted in purple, example partial assignments that are and are not satisfying the filter in green and red, respectively.}\label{fig:filter:ex}
\end{figure}

\Cref{fig:filter:ex} illustrates the idea behind each filtering predicate.
The cover of \(f\) is shown in purple, and the cover of valid assignments in green.
For the \(\exists\) predicate, we keep assignments for which there is at least one instance in \(f\) that is covered (\Cref{fig:filter:ex}a).
In the case of our \(\forall\) predicate, we eliminate assignments that do not completely cover \(f\) (\Cref{fig:filter:ex}b).
In contrast, filtering with predicate \({\supseteq}\) keeps assignments that are completely subsumed by the cover of \(f\) (\Cref{fig:filter:ex}c).

We first present \Cref{algo:eFilter} to filter a set of assignments based on the \(\exists\) predicate.
Given a meta-product BDD \(\mathcal{P}\) such that \(\llbracket{\mathcal{P}}\rrbracket \equiv m_{P}\) for some \(P \subseteq \pAs\), and a BDD \(\mathcal{F}\), we traverse \(\mathcal{P}\) and trim branches not satisfying the \(\exists\) criteria.
Along a path starting from the root meta-product BDD to a \(\textbf{1}\) terminal, we update \(\mathcal{F}\) accordingly. 
In particular, at a given node in \(\mathcal{P}\) associated with variable \(x_{i}\), whenever \(x_{i}\) does not occur in the assignment, i.e.\ \(o_{i}\) is \(\boolFalse\), the cover of the assignment may intersect with either the positive or negative cofactors of \(\mathcal{F}\), and thus we recurse on both sides and union both resulting sets.
Otherwise, if our assignment contains \(\overline{x_{i}}\) we consider the negative cofactor of \(\mathcal{F}\), and if our assignment contains \(\overline{x_{i}}\) we recurse on the postive cofactor of \(\mathcal{F}\).

\newcommand{\eFilter}{\ensuremath{\textsc{Filter}_{\exists}}}
\begin{algorithm}
\caption{$\eFilter(\mathcal{P},\mathcal{F})$}\label{algo:eFilter}
\textbf{Input: } meta-product BDD $\mathcal{P} = (T,V,\lambda,\delta,\rho)$, BDD $\mathcal{F} = (T',V',\lambda',\delta',\rho')$ \\
\textbf{Output: } meta-product BDD $\mathcal{Q}$
\begin{algorithmic}[1]
\If{$\llbracket{\mathcal{P}}\rrbracket \equiv \boolFalse$ or $\llbracket{\mathcal{F}}\rrbracket \equiv \boolFalse$} \Return{$\bddFalse$}
\ElsIf{$\llbracket{\mathcal{F}}\rrbracket \equiv \boolTrue$} \Return{$\mathcal{P}$}
\ElsIf{$\llbracket{\mathcal{P}}\rrbracket \equiv \boolTrue$}
    \State{\(x_{j} \leftarrow \lambda'(\rho')\)}
    \State{\Return{$\begin{aligned}[t]
    &(\overline{o_j} \land \eFilter(\mathcal{P}, \mathcal{F}_{x_j})) \lor (\overline{o_j} \land \eFilter(\mathcal{P}, \mathcal{F}_{\overline{x_j}})) \lor {}\\
    &(o_j \land s_j \land \eFilter(\mathcal{P}, \mathcal{F}_{x_j})) \lor {}\\
    &(o_j \land \overline{s_j} \land \eFilter(\mathcal{P}, \mathcal{F}_{\overline{x_j}}))
    \end{aligned}$}}
\Else{}
\State{$x_i \leftarrow \mathsf{var}(\lambda(\rho))$}
\State{$x_j \leftarrow \lambda'(\rho')$}
\If{$i = j$}
    \State{\Return{$\begin{aligned}[t]
    &(\overline{o_i} \land \eFilter(\mathcal{P}_{\overline{o_i}}, \mathcal{F}_{x_j})) \lor (\overline{o_i} \land \eFilter(\mathcal{P}_{\overline{o_i}}, \mathcal{F}_{\overline{x_j}})) \lor {}\\
    &(o_i \land s_i \land \eFilter(\mathcal{P}_{o_{i}s_{i}}, \mathcal{F}_{x_j})) \lor {}\\
    &(o_i \land \overline{s_i} \land \eFilter(\mathcal{P}_{o_{i}\overline{s_{i}}}, \mathcal{F}_{\overline{x_j}}))
    \end{aligned}$}}
\ElsIf{$i < j$}
    \State{\Return{$\begin{aligned}[t]
    &(o_i \land s_i \land \eFilter(\mathcal{P}_{o_{i}s_{i}}, \mathcal{F})) \lor {}\\
    &(o_i \land \overline{s_i} \land \eFilter(\mathcal{P}_{o_{i}\overline{s_{i}}}, \mathcal{F})) \lor {}\\
    &(\overline{o_i} \land \eFilter(\mathcal{P}_{\overline{o_i}}, \mathcal{F}))
    \end{aligned}$}}
\ElsIf{$i > j$}
    \State{\Return{$\begin{aligned}[t]
    &(\overline{o_j} \land \eFilter(\mathcal{P}, f_{x_j})) \lor (\overline{o_j} \land \eFilter(\mathcal{P}, \mathcal{F}_{\overline{x_j}})) \lor {}\\
    &(o_j \land s_j \land \eFilter(\mathcal{P}, \mathcal{F}_{x_j})) \lor {}\\
    &(o_j \land \overline{s_j} \land \eFilter(\mathcal{P}, \mathcal{F}_{\overline{x_j}}))
    \end{aligned}$}}
\EndIf{}
\EndIf{}
\end{algorithmic}
\end{algorithm}

\begin{restatable}{theorem}{pExist}
Let \(\mathcal{P}\) be the meta-product BDD of a set of assignments \(P \subseteq \pAs\) and BDD \(\mathcal{F}\) s.t.\ \(\llbracket{\mathcal{F}}\rrbracket \equiv f\) and \(\llbracket{f}\rrbracket \subseteq \tAs\).
\(\eFilter(\mathcal{P},\mathcal{F})\) is the meta-product of the largest set of assignments \(Q \subseteq P\) where for all \(\mathbf{q} \in Q\), there exists a \(\mathbf{u} \in \llbracket{f}\rrbracket\) such that \(\mathbf{u} \in \llbracket{\mathbf{q}}\rrbracket\).
\end{restatable}
{
\allowdisplaybreaks
\begin{proof}
Given a meta-product BDD \(\mathcal{P}\) encoding \(P \subseteq \pAs\) and BDD \(\llbracket{\mathcal{F}}\rrbracket \equiv f\) for Boolean function \(f\), we show that \(\eFilter(\mathcal{P},\mathcal{F}) \equiv \mathcal{M}_{Q}\) where:
\begin{align}
(\forall \mathbf{q} \in Q.\ \exists \mathbf{u} \in \llbracket{f}\rrbracket.\ \mathbf{u} \in \llbracket{\mathbf{q}}\rrbracket) &\land {}\label{proof:eq:1:1}\\
(\forall \mathbf{r} \in P \setminus Q.\ \forall \mathbf{w} \in \llbracket{f}\rrbracket.\ \mathbf{w} \not\in \llbracket{\mathbf{r}}\rrbracket)&\label{proof:eq:1:2}
\end{align}

We prove by structural induction on \(\mathcal{P}\).
If \(\llbracket{\mathcal{P}}\rrbracket \equiv \boolFalse\), it represents the empty set of assignments, and thus we return the empty set representation (\(\bddFalse\)), and our claim holds.
If \(\mathcal{P}\) is \(\boolTrue\) (lines 3--5), we perform induction on \(\mathcal{F}\):
\begin{itemize}[align=left,leftmargin=0pt,itemindent=*]
\item[\itshape Base Case (\(f\) is constant):] If \(f \equiv \llbracket{\mathcal{F}}\rrbracket \equiv \boolFalse\), we return the empty set representation, and our claim holds. 
If \(f \equiv \llbracket{\mathcal{F}}\rrbracket \equiv \boolTrue\), we return the largest set satisfying our claim, and thus we return \(\bddTrue\).
\item[\itshape Step Case (\(f\) is not constant):] \sloppy Let \(\mathcal{F}\) be rooted by variable \(x_{j}\).
Assume the induction hypothesis that our claim holds for \(\eFilter(\boolTrue,\mathcal{F}_{\overline{x_{j}}})\) and \(\eFilter(\boolTrue,\mathcal{F}_{x_{j}})\), and let \(Q_{\overline{x_{j}}}\) and \(Q_{x_{j}}\) be the resulting sets, respectively.
We derive the following: 
{\setlength{\abovedisplayskip}{2pt}\setlength{\belowdisplayskip}{0pt}
\begin{align*}
& \forall \mathbf{q} \in {Q}_{\overline{x_{j}}}.\ \exists \mathbf{u} \in \llbracket{f_{\overline{x_j}}}\rrbracket.\ \mathbf{u} \in \llbracket{\mathbf{q}}\rrbracket\\
\Longleftrightarrow\ &\\
& \forall \mathbf{q} \in {Q}_{\overline{x_{j}}}.\ \exists \mathbf{u}[0/x_{j}] \in \llbracket{f}\rrbracket.\ \mathbf{u}[0/x_{j}] \in \llbracket{\mathbf{q}}\rrbracket\\
\Longrightarrow\ &\\
& \forall \mathbf{q} \in \{\mathbf{q'} \mid \mathbf{q'}[0/x_{j}] \in{Q}_{\overline{x_{j}}}\} \cup \{\mathbf{q'} \mid \mathbf{q'}[*/x_{j}] \in{Q}_{\overline{x_{j}}}\}.\ \exists \mathbf{u} \in \llbracket{f}\rrbracket.\ \mathbf{u} \in \llbracket{\mathbf{q}}\rrbracket\\
\end{align*}\\[-2em]
and
\begin{align*}
& \forall \mathbf{q} \in {Q}_{x_{j}}.\ \exists \mathbf{u} \in \llbracket{f_{x_j}}\rrbracket.\ \mathbf{u} \in \llbracket{\mathbf{q}}\rrbracket\\
\Longleftrightarrow\ &\\
& \forall \mathbf{q} \in {Q}_{x_{j}}.\ \exists \mathbf{u}[1/x_{j}] \in \llbracket{f}\rrbracket.\ \mathbf{u}[1/x_{j}] \in \llbracket{\mathbf{q}}\rrbracket\\
\Longrightarrow\ &\\
& \forall \mathbf{q} \in \{\mathbf{q'} \mid \mathbf{q'}[1/x_{j}] \in{Q}_{x_{j}}\}\cup\{\mathbf{q'} \mid \mathbf{q'}[*/x_{j}] \in{Q}_{x_{j}}\}.\ \exists \mathbf{u} \in \llbracket{f}\rrbracket.\ \mathbf{u} \in \llbracket{\mathbf{q}}\rrbracket\\
\end{align*}
}
In line 5, we return the set \(Q\) defined as follows:
{\setlength{\abovedisplayskip}{2pt}\setlength{\belowdisplayskip}{2pt}
\begin{align*}
    &\{\mathbf{q} \mid \mathbf{q}[*/x_{j}]  \in Q_{\overline{x_{j}}}\} \cup \{\mathbf{q} \mid \mathbf{q}[*/x_{j}] \in Q_{x_{j}}\} \cup {}\\
    &\{\mathbf{q} \mid \mathbf{q}[0/x_{j}] \in Q_{\overline{x_{j}}}\} \cup \{\mathbf{q} \mid \mathbf{q}[1/x_{j}]  \in Q_{x_{j}}\}
\end{align*}}%
By our IH, we have that \(Q\) satisfies \Cref{proof:eq:1:1}.
Moreover, by our IH, we have that Q satisfies \Cref{proof:eq:1:2}.
\begin{align*}
& (\forall \mathbf{r} \in \overline{{Q}_{\overline{x_{j}}}}.\ \forall \mathbf{w} \in \llbracket{f_{\overline{x_{j}}}}\rrbracket.\ \mathbf{w} \not\in \llbracket{\mathbf{r}}\rrbracket) \land (\forall \mathbf{r} \in \overline{{Q}_{x_{j}}}.\ \forall \mathbf{w} \in \llbracket{f_{x_j}}\rrbracket.\ \mathbf{w} \not\in \llbracket{\mathbf{r}}\rrbracket)\\
\Longleftrightarrow\ &\\
& (\forall \mathbf{r} \in \overline{{Q}_{\overline{x_{j}}}}.\ \forall \mathbf{w}[0/x_{j}] \in \llbracket{f}\rrbracket.\ \mathbf{w}[0/x_{j}] \not\in \llbracket{\mathbf{r}}\rrbracket) \land {}\\
& \qquad (\forall \mathbf{r} \in \overline{{Q}_{x_{j}}}.\ \forall \mathbf{w}[1/x_{j}] \in \llbracket{f}\rrbracket.\ \mathbf{w}[1/x_{j}] \not\in \llbracket{\mathbf{r}}\rrbracket)\\
\Longleftrightarrow\ & \{\text{Property }\mathbf{w}[0/x] \notin \llbracket{\mathbf{r}[1/x]}\rrbracket \land \mathbf{w}[1/x] \notin \llbracket{\mathbf{r}[0/x]}\rrbracket\}\\
& (\forall \mathbf{r} \in \overline{\{\mathbf{s} \mid \mathbf{s}[*/x_{j}] \in Q_{\overline{x_{j}}} \lor \mathbf{s}[0/x_{j}] \in Q_{\overline{x_{j}}}\}}.\ \forall \mathbf{w}[0/x_{j}] \in \llbracket{f}\rrbracket.\ \mathbf{w}[0/x_{j}] \not\in \llbracket{\mathbf{r}}\rrbracket) \land {}\\
& (\forall \mathbf{r} \in \overline{\{\mathbf{s} \mid \mathbf{s}[*/x_{j}] \in Q_{x_{j}} \lor \mathbf{s}[1/x_{j}] \in Q_{x_{j}}\}}.\ \forall \mathbf{w}[1/x_{j}] \in \llbracket{f}\rrbracket.\ \mathbf{w}[1/x_{j}] \not\in \llbracket{\mathbf{r}}\rrbracket)\\
\Longleftrightarrow\ &\\
& \forall \mathbf{r} \in \overline{Q}.\ \forall \mathbf{w} \in \llbracket{f}\rrbracket.\ \mathbf{w} \not\in \llbracket{\mathbf{r}}\rrbracket
\end{align*}
\end{itemize}

Let \(\mathcal{P}\) be rooted by variable \(o_{i}\).
Assume the induction hypothesis that our claim holds for \(\eFilter(\mathcal{P}_{\overline{o_{i}}},\mathcal{F})\), \(\eFilter(\mathcal{P}_{o_{i}\overline{s_{i}}},\mathcal{F})\), and \(\eFilter(\mathcal{P}_{o_{i}s_{i}}, \mathcal{F})\).
If \(\llbracket{\mathcal{F}}\rrbracket \equiv \boolFalse\), we return the empty set representation, and our claim holds. If \(\llbracket{\mathcal{F}}\rrbracket \equiv \boolTrue\), we return the largest set satisfying our claim, \(\mathcal{P}\).
Otherwise, let \(\mathcal{F}\) be rooted by variable \(x_{j}\).
We now perform a case distinction:
\begin{itemize}[align=left,leftmargin=20pt,itemindent=*,labelindent=20pt]
\item[\itshape Case: (\(i = j\))] \sloppy
    Let us define \(Q_{*0}, Q_{*1}, Q_{0}, Q_{1}, P_{*}, P_{0}, P_{1}\) such that: 
{
\setlength{\abovedisplayskip}{0pt}
\setlength{\belowdisplayskip}{2pt}
    \begin{align*}
    \mathcal{M}_{Q_{*0}} &\equiv \llbracket \eFilter(\mathcal{P}_{\overline{o_{i}}},\mathcal{F}_{\overline{x_{j}}})  \rrbracket
    &
    \mathcal{M}_{Q_{0}} &\equiv \llbracket \eFilter(\mathcal{P}_{o_{i}\overline{s_{i}}},\mathcal{F}_{\overline{x_{j}}}) \rrbracket
    \\
    \mathcal{M}_{Q_{*1}} &\equiv \llbracket \eFilter(\mathcal{P}_{\overline{o_{i}}},\mathcal{F}_{x_{j}}) \rrbracket
    &
    \mathcal{M}_{Q_{1}} &\equiv \llbracket \eFilter(\mathcal{P}_{o_{i}s_{i}},\mathcal{F}_{x_{j}}) \rrbracket
    \\[1em]
    P_{*} &= \{\mathbf{p} \mid \mathbf{p}[*/x_{i}] \in P\}
    &
    P_{0} &= \{\mathbf{p} \mid \mathbf{p}[0/x_{i}] \in P\} 
    \\
    &&
    P_{1} &= \{\mathbf{p} \mid \mathbf{p}[1/x_{i}] \in P\}
    \end{align*}
}
    Moreover, let \(Q\) be defined as 
    \[
    \begin{gathered}
        \{\mathbf{q} \mid \mathbf{q}[*/x_{i}] \in Q_{*0}\} \cup \{\mathbf{q} \mid \mathbf{q}[*/x_{i}] \in Q_{*1}\} \cup {}\\
        \{\mathbf{q} \mid \mathbf{q}[0/x_{i}] \in Q_{0}\} \cup \{\mathbf{q} \mid \mathbf{q}[1/x_{i}] \in Q_{1}\}
    \end{gathered}
    \]
    By the induction hypothesis, we have the following:
    \begin{align*}
        & (\forall \mathbf{q} \in Q_{*0}.\ \exists \mathbf{u} \in \llbracket{f_{\overline{x_{i}}}}\rrbracket.\ \mathbf{u} \in \llbracket{\mathbf{q}}\rrbracket) \land {}
        (\forall \mathbf{q} \in Q_{*1}.\ \exists \mathbf{u} \in \llbracket{f_{x_{i}}}\rrbracket.\ \mathbf{u} \in \llbracket{\mathbf{q}}\rrbracket) \land {}\\
        & (\forall \mathbf{q} \in Q_{0}.\ \exists \mathbf{u} \in \llbracket{f_{\overline{x_{i}}}}\rrbracket.\ \mathbf{u} \in \llbracket{\mathbf{q}}\rrbracket) \land {}
        (\forall \mathbf{q} \in Q_{1}.\ \exists \mathbf{u} \in \llbracket{f_{x_{i}}}\rrbracket.\ \mathbf{u} \in \llbracket{\mathbf{q}}\rrbracket)\\
        \Longleftrightarrow\ &\\
        & (\forall \mathbf{q} \in Q_{*0}.\ \exists \mathbf{u}[0/x_{i}] \in \llbracket{f}\rrbracket.\ \mathbf{u}[0/x_{i}] \in \llbracket{\mathbf{q}}\rrbracket) \land {}\\
        & (\forall \mathbf{q} \in Q_{*1}.\ \exists \mathbf{u}[1/x_{i}] \in \llbracket{f}\rrbracket.\ \mathbf{u}[1/x_{i}] \in \llbracket{\mathbf{q}}\rrbracket) \land {}\\
        & (\forall \mathbf{q} \in Q_{0}.\ \exists \mathbf{u}[0/x_{i}] \in \llbracket{f}\rrbracket.\ \mathbf{u}[0/x_{i}] \in \llbracket{\mathbf{q}}\rrbracket) \land {}\\
        & (\forall \mathbf{q} \in Q_{1}.\ \exists \mathbf{u}[1/x_{i}] \in \llbracket{f}\rrbracket.\ \mathbf{u}[1/x_{i}] \in \llbracket{\mathbf{q}}\rrbracket)\\
        \Longrightarrow\ &\\
        & (\forall \mathbf{q} \in \{\mathbf{q'} | \mathbf{q'}[0/x_{i}] \in Q_{*0}\}.\ \exists \mathbf{u}[0/x_{i}] \in \llbracket{f}\rrbracket.\ \mathbf{u}[0/x_{i}] \in \llbracket{\mathbf{q}}\rrbracket) \land {}\\
        & (\forall \mathbf{q} \in \{\mathbf{q'} | \mathbf{q'}[1/x_{i}] \in Q_{*1}\}.\ \exists \mathbf{u}[1/x_{i}] \in \llbracket{f}\rrbracket.\ \mathbf{u}[1/x_{i}] \in \llbracket{\mathbf{q}}\rrbracket) \land {}\\
        & (\forall \mathbf{q} \in \{\mathbf{q'} | \mathbf{q'}[0/x_{i}] \in Q_{0}\}.\ \exists \mathbf{u}[0/x_{i}] \in \llbracket{f}\rrbracket.\ \mathbf{u}[0/x_{i}] \in \llbracket{\mathbf{q}}\rrbracket) \land {}\\
        & (\forall \mathbf{q} \in \{\mathbf{q'} | \mathbf{q'}[1/x_{i}] \in Q_{1}\}.\ \exists \mathbf{u}[1/x_{i}] \in \llbracket{f}\rrbracket.\ \mathbf{u}[1/x_{i}] \in \llbracket{\mathbf{q}}\rrbracket)\\
        \Longleftrightarrow\ &\\
        & (\forall \mathbf{q} \in \{\mathbf{q'} | \mathbf{q'}[0/x_{i}] \in Q_{*0}\}.\ \exists \mathbf{u} \in \llbracket{f}\rrbracket.\ \mathbf{u} \in \llbracket{\mathbf{q}}\rrbracket) \land {}\\
        & (\forall \mathbf{q} \in \{\mathbf{q'} | \mathbf{q'}[1/x_{i}] \in Q_{*1}\}.\ \exists \mathbf{u} \in \llbracket{f}\rrbracket.\ \mathbf{u} \in \llbracket{\mathbf{q}}\rrbracket) \land {}\\
        & (\forall \mathbf{q} \in \{\mathbf{q'} | \mathbf{q'}[0/x_{i}] \in Q_{0}\}.\ \exists \mathbf{u} \in \llbracket{f}\rrbracket.\ \mathbf{u} \in \llbracket{\mathbf{q}}\rrbracket) \land {}\\
        & (\forall \mathbf{q} \in \{\mathbf{q'} | \mathbf{q'}[1/x_{i}] \in Q_{1}\}.\ \exists \mathbf{u} \in \llbracket{f}\rrbracket.\ \mathbf{u} \in \llbracket{\mathbf{q}}\rrbracket)\\
    \end{align*}
    It follows that \(Q\) satisfies \Cref{proof:eq:1:2}.
    To prove that \(Q\) satisfies \Cref{proof:eq:1:2}, we reason as follows:
    \begin{align*}
        & (\forall \mathbf{r} \in P_* \setminus Q_{*0}.\ \forall\mathbf{w} \in \llbracket{f_{\overline{x_{i}}}}\rrbracket.\ \mathbf{w} \not\in \llbracket{\mathbf{r}}\rrbracket) \land {}\\
        & (\forall \mathbf{r} \in P_* \setminus Q_{*1}.\ \forall \mathbf{w} \in \llbracket{f_{x_{i}}}\rrbracket.\ \mathbf{w} \not\in \llbracket{\mathbf{r}}\rrbracket) \land {}\\
        & (\forall \mathbf{r} \in P_{0} \setminus Q_{0}.\ \forall \mathbf{w} \in \llbracket{f_{\overline{x_{i}}}}\rrbracket.\ \mathbf{w} \not\in \llbracket{\mathbf{r}}\rrbracket) \land {}\\
        & (\forall \mathbf{r} \in P_{1} \setminus Q_{1}.\ \forall \mathbf{w} \in \llbracket{f_{x_{i}}}\rrbracket.\ \mathbf{w} \not\in \llbracket{\mathbf{r}}\rrbracket)\\
        \Longleftrightarrow\ &\\
        & (\forall \mathbf{r} \in P_* \setminus Q_{*0}.\ \forall\mathbf{w}[0/x_{i}] \in \llbracket{f}\rrbracket.\ \mathbf{w}[0/x_i] \not\in \llbracket{\mathbf{r}}\rrbracket) \land {}\\
        & (\forall \mathbf{r} \in P_* \setminus Q_{*1}.\ \forall \mathbf{w}[1/x_i] \in \llbracket{f}\rrbracket.\ \mathbf{w}[1/x_i] \not\in \llbracket{\mathbf{r}}\rrbracket) \land {}\\
        & (\forall \mathbf{r} \in P_{0} \setminus Q_{0}.\ \forall \mathbf{w}[0/x_i] \in \llbracket{f}\rrbracket.\ \mathbf{w}[0/x_i] \not\in \llbracket{\mathbf{r}}\rrbracket) \land {}\\
        & (\forall \mathbf{r} \in P_{1} \setminus Q_{1}.\ \forall \mathbf{w}[1/x_i] \in \llbracket{f}\rrbracket.\ \mathbf{w}[1/x_i] \not\in \llbracket{\mathbf{r}}\rrbracket)\\
        \Longleftrightarrow\ &\{\text{Property }\mathbf{w}[0/x] \notin \llbracket{\mathbf{r}[1/x]}\rrbracket \land \mathbf{w}[1/x] \notin \llbracket{\mathbf{r}[0/x]}\rrbracket\}\\
        & (\forall \mathbf{r} \in P_* \setminus \{\mathbf{s} \mid \mathbf{s}[*/x_i] \in Q_{*0} \lor \mathbf{s}[0/x_i] \in Q_{*0}\}.\ \forall\mathbf{w}[0/x_{i}] \in \llbracket{f}\rrbracket.\ \mathbf{w}[0/x_i] \not\in \llbracket{\mathbf{r}}\rrbracket) \land {}\\
        & (\forall \mathbf{r} \in P_* \setminus \{\mathbf{s} \mid \mathbf{s}[*/x_i] \in Q_{*1} \lor \mathbf{s}[1/x_i] \in Q_{*1}\}.\ \forall \mathbf{w}[1/x_i] \in \llbracket{f}\rrbracket.\ \mathbf{w}[1/x_i] \not\in \llbracket{\mathbf{r}}\rrbracket) \land {}\\
        & (\forall \mathbf{r} \in P_{0} \setminus \{\mathbf{s} \mid \mathbf{s}[*/x_i] \in Q_{0} \lor \mathbf{s}[0/x_i] \in Q_{0}\}.\ \forall \mathbf{w}[0/x_i] \in \llbracket{f}\rrbracket.\ \mathbf{w}[0/x_i] \not\in \llbracket{\mathbf{r}}\rrbracket) \land {}\\
        & (\forall \mathbf{r} \in P_{1} \setminus \{\mathbf{s} \mid \mathbf{s}[*/x_i] \in Q_{1} \lor \mathbf{s}[1/x_i] \in Q_{1}\}.\ \forall \mathbf{w}[1/x_i] \in \llbracket{f}\rrbracket.\ \mathbf{w}[1/x_i] \not\in \llbracket{\mathbf{r}}\rrbracket)\\
        \Longleftrightarrow\ &\{\text{Definition of } P_{*}, P_{0}, P_{1}\}\\
        & (\forall \mathbf{r} \in P_* \setminus \{\mathbf{s} \mid \mathbf{s}[*/x_i] \in Q_{*0}\}.\ \forall\mathbf{w}[0/x_{i}] \in \llbracket{f}\rrbracket.\ \mathbf{w}[0/x_i] \not\in \llbracket{\mathbf{r}}\rrbracket) \land {}\\
        & (\forall \mathbf{r} \in P_* \setminus \{\mathbf{s} \mid \mathbf{s}[*/x_i] \in Q_{*1}\}.\ \forall \mathbf{w}[1/x_i] \in \llbracket{f}\rrbracket.\ \mathbf{w}[1/x_i] \not\in \llbracket{\mathbf{r}}\rrbracket) \land {}\\
        & (\forall \mathbf{r} \in P_{0} \setminus \{\mathbf{s} \mid \mathbf{s}[0/x_i] \in Q_{0}\}.\ \forall \mathbf{w}[0/x_i] \in \llbracket{f}\rrbracket.\ \mathbf{w}[0/x_i] \not\in \llbracket{\mathbf{r}}\rrbracket) \land {}\\
        & (\forall \mathbf{r} \in P_{1} \setminus \{\mathbf{s} \mid \mathbf{s}[1/x_i] \in Q_{1}\}.\ \forall \mathbf{w}[1/x_i] \in \llbracket{f}\rrbracket.\ \mathbf{w}[1/x_i] \not\in \llbracket{\mathbf{r}}\rrbracket)\\
        \Longleftrightarrow\ &\{\text{Definition of } P \text{ and } Q\}\\
        & \forall \mathbf{r} \in P \setminus Q.\ \forall \mathbf{w} \in \llbracket{f}\rrbracket.\ \mathbf{w} \not\in \llbracket{\mathbf{r}}\rrbracket
    \end{align*}
    \item[\itshape Cases (\(i > j\)) and (\(i < j\)):] These cases follows a similar argumentation. \qed
\end{itemize}
\end{proof}

}

\newcommand{\aFilter}{\ensuremath{\textsc{Filter}_{\forall}}}
\begin{algorithm}
\caption{$\aFilter(\mathcal{P}, \mathcal{F})$}\label{algo:aFilter}
\textbf{Input: } meta-product BDD $\mathcal{P} = (T, V, \lambda, \delta, \rho)$, BDD $\mathcal{F} = (T', V', \lambda', \delta', \rho')$ \\
\textbf{Output: } meta-product BDD $\mathcal{Q}$
\begin{algorithmic}[1]
\If{$\llbracket\mathcal{P}\rrbracket \equiv \boolFalse$ or $\llbracket\mathcal{F}\rrbracket \equiv \boolFalse$} \Return{$\mathcal{P}$}
\ElsIf{$\llbracket\mathcal{P}\rrbracket = \boolTrue \land \llbracket\mathcal{F}\rrbracket = \boolTrue$} \Return{$\bddTrue$}
\ElsIf{$\llbracket\mathcal{P}\rrbracket = \boolTrue$}
\State{$x_{j} \leftarrow \lambda'(\rho')$}
\State{\Return{$\begin{aligned}[t]
&(\overline{o_j} \land \aFilter(\mathcal{P}, \mathcal{F}_{x_j}) \land \aFilter(\mathcal{P}, \mathcal{F}_{\overline{x_j}}))
\end{aligned}$}}
\ElsIf{$\llbracket\mathcal{F}\rrbracket \equiv \boolTrue$} 
\State{$x_i \leftarrow \mathsf{var}(\lambda(\rho))$}
\State{\Return{$(\overline{o_i} \land \aFilter(\mathcal{P}_{\overline{o_i}}, \mathcal{F}))$}}
\Else
\State{$x_i \leftarrow \mathsf{var}(\lambda(\rho))$}
\State{$x_j \leftarrow \lambda'(\rho')$}
\If{$i = j$}
    \If{$(\llbracket\mathcal{F}_{x_i}\rrbracket \equiv \boolFalse)$}
    \State{\Return{$\begin{aligned}[t]
        &({o_i} \land \overline{s_i} \land \aFilter(\mathcal{P}_{{o_i}{s_{i}}}, \mathcal{F}_{\overline{x_i}})) \lor {}\\
        &(\overline{o_i} \land \aFilter(\mathcal{P}_{\overline{o_{i}}}, \mathcal{F}_{x_j}) \land \aFilter(\mathcal{P}_{\overline{o_{i}}}, \mathcal{F}_{\overline{x_j}}))
        \end{aligned}$}}
    \ElsIf{$(\llbracket\mathcal{F}_{\overline{x_i}}\rrbracket \equiv \boolFalse)$}
    \State{\Return{$\begin{aligned}[t]
        &({o_i} \land {s_i} \land \aFilter(\mathcal{P}_{{o_i}{s_{i}}}, \mathcal{F}_{x_i})) \lor {}\\
        &(\overline{o_i} \land \aFilter(\mathcal{P}_{\overline{o_{i}}}, \mathcal{F}_{x_j}) \land \aFilter(\mathcal{P}_{\overline{o_{i}}}, \mathcal{F}_{\overline{x_j}}))
        \end{aligned}$}}
    \Else{} {\Return{$\begin{aligned}[t]
        &(\overline{o_i} \land \aFilter(\mathcal{P}_{\overline{o_{i}}}, \mathcal{F}_{x_j}) \land \aFilter(\mathcal{P}_{\overline{o_{i}}}, \mathcal{F}_{\overline{x_j}}))
        \end{aligned}$}}
    \EndIf{}
\ElsIf{$i < j$}
    \State{\Return{$\begin{aligned}[t]
    &(\overline{o_i} \land \aFilter(\mathcal{P}_{\overline{o_i}}, \mathcal{F}))
    \end{aligned}$}}
\ElsIf{$i > j$}
    \State{\Return{$\begin{aligned}[t]
        &(\overline{o_j} \land \aFilter(\mathcal{P}, \mathcal{F}_{x_j}) \land \aFilter(\mathcal{P}, \mathcal{F}_{\overline{x_j}}))
    \end{aligned}$}}
\EndIf{}
\EndIf{}
\end{algorithmic}
\end{algorithm}

In \Cref{algo:aFilter}, we universally filter the meta-product BDD \(\llbracket{\mathcal{P}}\rrbracket \equiv m_{P}\) w.r.t.\ Boolean function \(f\) using BDD \(\llbracket{\mathcal{F}}\rrbracket \equiv f\).
The algorithm operates in a similar vein to \Cref{algo:eFilter}.
If \(f\) is \(\boolFalse\), we vacuously satisfy the predicate \({\forall}\) for all \(\mathbf{p} \in P\), and thus we simply return \(\mathcal{P}\).
Otherwise, at a given node of \(\mathcal{P}\) with variable associated with \(x_{i}\), whenever \(o_{i}\) is \(\boolFalse\), the assignment we consider will not have variable \(x_{i}\).
Since we must cover \(f\) completely, the assignment must be present in the sub-calls in both cofactors of \(f\).
Otherwise, we recurse on the cofactors of \(f\) depending on the setting of \(x_{i}\) in our assignment.
However, we must not recurse on a cofactor of \(f\) if it immediately leads to \(\boolFalse\) (lines 13--16), as the base case \(f \equiv \boolFalse\) implies that \(\mathcal{P}\) immediately covers the empty set; thus, to trim irrelevant branch we must perform a check before further recursions.
\begin{restatable}{theorem}{pAll}
Let \(\mathcal{P}\) be the meta-product BDD of a set of assignments \(P \subseteq \pAs\) and BDD \(\mathcal{F}\) s.t.\ \(\llbracket{\mathcal{F}}\rrbracket \equiv f\) and \(\llbracket{f}\rrbracket \subseteq \tAs\).
\(\aFilter(\mathcal{P},\mathcal{F})\) is the meta-product of the largest set of assignments \(Q \subseteq P\) where for all \(\mathbf{q} \in Q\) and for all \(\mathbf{u} \in \llbracket{f}\rrbracket\), we have that \(\mathbf{u} \in \llbracket{\mathbf{q}}\rrbracket\).
\end{restatable}
{
\allowdisplaybreaks
\begin{proof}[Sketch]
The proof follows from structural induction on \(\mathcal{P}\) and \(\mathcal{F}\).
For brevity, we offer a proof sketch.
Let \(\llbracket{\mathcal{F}}\rrbracket \equiv f\) for Boolean function \(f\).
If \(\llbracket{\mathcal{P}}\rrbracket \equiv m_{\emptyset} \equiv \boolFalse\) or \(\llbracket{\mathcal{F}}\rrbracket \equiv \boolFalse \) (i.e.\ \(\llbracket{f}\rrbracket = \emptyset\)), we correctly return \(\bddFalse\) since \(\llbracket{\bddFalse}\rrbracket \equiv m_{\emptyset}\).
Only if \(\llbracket{\mathcal{P}}\rrbracket \equiv m_{\pAs} \equiv \boolTrue\) and \(\llbracket{\mathcal{F}}\rrbracket \equiv \boolTrue\) (i.e.\ \(\llbracket{f}\rrbracket = \pAs\)) we wish to return \(\pAs\) as this is the largest set satisfying our claim, i.e.\ we explicitly return \(\bddTrue\) since \(\llbracket{\bddTrue}\rrbracket \equiv m_{\pAs}\).
Otherwise, we recurse on the BDD substructures.

Observe that whenever \(f \equiv \boolTrue\), we must only consider assignments for which we do not restrict the covered space further, otherwise we no longer subsume \(f\) (lines 6--8).
Moreover, if \(\llbracket{\mathcal{P}}\rrbracket \equiv m_{\pAs}\) and \(f\) is not constant, we also only consider assignments that do not restrict the covered space, and thus only consider assignments where \(o_{j}\) is \(\boolFalse\), where \(x_{j}\) is the root of \(f\).

Now consider the case where \(\llbracket{\mathcal{P}}\rrbracket \not\equiv m_{\emptyset}\) and \(\llbracket{\mathcal{F}}\rrbracket \not\equiv \boolFalse\), i.e.\ \(\llbracket{f}\rrbracket \neq \emptyset\).
Let \(x_{i}\) be the variable associated with the root label of \(\mathcal{P}\) and \(x_{j}\) the variable rooted at \(\mathcal{F}\).
If the variable \(x_{i}\) occurs before \(x_{j}\), i.e.\ \(i < j\) (lines 18--19), we must only consider the assignments in \(\mathcal{P}\) where \(x_{i}\) does not occur. 
This is because \(x_{i}\) does not impact the evaluation of \(f\) and by having assignments restricting \(x_{i}\), we will no longer fully cover \(f\).
If \(i > j\), we again consider non-occurrence of \(x_{j}\) since assignments that contain \(x_{j}\) will no longer subsume both \(f_{x_{j}}\) and \(f_{\overline{x_{j}}}\).

Finally, we consider the case where \(i = j\).
Since in the base cases we have that whenever \(\llbracket{F}\rrbracket \equiv \boolFalse\) we return \(\mathcal{P}\), we must perform a check before recursion.
If the positive cofactor of \(\mathcal{F}\) is \(\mathbf{0}\), i.e.\ \(f_{x} \equiv \boolFalse\), the space that \(f\) covers must only have \(x_{j}\) set to \(0\). 
Thus, we must consider the assignments that set \(x_{j}\) to \(0\) and the case where \(x_{j}\) does not occur. 
The same argument holds for whenever the negative cofactor of \(\mathcal{F}\) is \(\mathbf{0}\), i.e.\ \(f_{\overline{x}} \equiv \boolFalse\).
Otherwise, we cannot choose assignments that restrict the cover, otherwise we no longer fully cover \(f\).\qed
\end{proof}

}

Finally, we present \Cref{algo:cFilter} for filtering \(\mathcal{P} \equiv m_{P}\) according to \({\supseteq}\) w.r.t.\ BDD \(\mathcal{F}\), where \(\llbracket{\mathcal{F}}\rrbracket \equiv f\).
If \(f \equiv \boolFalse\), no assignment satisfies our predicate filter, and thus we return the empty set representation \(\bddFalse\) (line 1).
Moreover, if \(f \equiv \boolTrue\), any assignment in \(P\) satisfies the filtering predicate, and we return \(\mathcal{P}\) (line 2).
Because our assignment's cover must be completely contained within \(f\), whenever \(o_{i}\) is equivalent to \(\boolFalse\), i.e.\ \(x_{i}\) is not in our assignment, we must check that our assignment is completely covered that both cofactors of \(f\). 
Otherwise, we recurse on the cofactors of \(f\) based on the setting of \(x_{i}\) in our assignment.

\newcommand{\cFilter}{\ensuremath{\textsc{Filter}_{\supseteq}}}
\begin{algorithm}
\caption{$\cFilter(\mathcal{P},\mathcal{F})$}\label{algo:cFilter}
\textbf{Input: } meta-product BDD $\mathcal{P} = (T, V, \lambda, \delta, \rho)$, BDD $\mathcal{F} = (T', V', \lambda', \delta', \rho')$ \\
\textbf{Output: } meta-product BDD $\mathcal{Q}$
\begin{algorithmic}[1]
\If{$\llbracket\mathcal{P}\rrbracket \equiv \boolFalse$ or $\llbracket\mathcal{F}\rrbracket \equiv \boolFalse$} \Return{$\bddFalse$}
\ElsIf{$\llbracket\mathcal{F}\rrbracket \equiv \boolTrue$} \Return{$\mathcal{P}$}
\ElsIf{$\llbracket\mathcal{P}\rrbracket \equiv \boolTrue$} 
    \State{$x_{j} \leftarrow \lambda'(\rho')$}
    \State{\Return{$\begin{aligned}[t]
    &(o_i \land s_i \land \cFilter(\mathcal{P}, \mathcal{F}_{x_j})) \lor (o_i \land \overline{s_i} \land \cFilter(\mathcal{P}, \mathcal{F}_{\overline{x_j}})) \lor {}\\
    &(\overline{o_i} \land \cFilter(\mathcal{P}, \mathcal{F}_{x_j}) \land \cFilter(\mathcal{P}, \mathcal{F}_{\overline{x_j}}))
    \end{aligned}$}}
\Else{}
    \State{$x_i \leftarrow \mathsf{var}(\lambda(\rho))$}
    \State{$x_j \leftarrow \lambda'(\rho')$}
    \If{$i = j$}
        \State{\Return{$\begin{aligned}[t]
        &(o_i \land s_i \land \cFilter(\mathcal{P}_{o_{i}s_{i}}, \mathcal{F}_{x_j})) \lor (o_i \land \overline{s_i} \land \cFilter(\mathcal{P}_{o_{i}\overline{s_{i}}}, \mathcal{F}_{\overline{x_j}})) \lor {}\\
        &(\overline{o_i} \land \cFilter(\mathcal{P}_{\overline{o_{i}}}, \mathcal{F}_{x_j}) \land \cFilter(\mathcal{P}_{\overline{o_{i}}}, \mathcal{F}_{\overline{x_j}}))
        \end{aligned}$}}
    \ElsIf{$i < j$}
        \State{\Return{$\begin{aligned}[t]
        &(o_i \land s_i \land \cFilter(\mathcal{P}_{o_{i}s_{i}}, \mathcal{F})) \lor (o_i \land \overline{s_i} \land \cFilter(\mathcal{P}_{o_{i}\overline{s_{i}}}, \mathcal{F})) \lor {}\\
        &(\overline{o_i} \land \cFilter(\mathcal{P}_{\overline{o_i}}, \mathcal{F}))
        \end{aligned}$}}
    \ElsIf{$i > j$}
        \State{\Return{$\begin{aligned}[t]
        &(o_i \land s_i \land \cFilter(\mathcal{P}, \mathcal{F}_{x_j})) \lor (o_i \land \overline{s_i} \land \cFilter(\mathcal{P}, \mathcal{F}_{\overline{x_j}})) \lor {}\\
        &(\overline{o_i} \land \cFilter(\mathcal{P}, \mathcal{F}_{x_j}) \land \cFilter(\mathcal{P}, \mathcal{F}_{\overline{x_j}}))
        \end{aligned}$}}
    \EndIf{}
\EndIf{}
\end{algorithmic}
\end{algorithm}

\begin{restatable}{theorem}{pCov}
Let \(\mathcal{P}\) be the meta-product BDD of a set of assignments \(P \subseteq \pAs\) and \(\mathcal{F}\) be a BDD s.t.\ \(\llbracket{\mathcal{F}}\rrbracket \equiv f\) and \(\llbracket{f}\rrbracket \subseteq \tAs\).
\(\cFilter(\mathcal{P},\mathcal{F})\) is the meta-product of the largest set of assignments \(Q \subseteq P\) where for all \(\mathbf{q} \in Q\) and for all \(\mathbf{u} \in \llbracket{\mathbf{q}}\rrbracket\), we have that \(\mathbf{u} \in \llbracket{f}\rrbracket\).
\end{restatable}
{
\allowdisplaybreaks
\begin{proof}[Sketch]
The proof follows on structural induction on \(\mathcal{P}\) and \(\mathcal{F}\).
In the base case, whenever \(\llbracket{\mathcal{P}}\rrbracket \equiv m_{\emptyset}\) or \(\llbracket{\mathcal{F}}\rrbracket \equiv \boolFalse\) (i.e.\ \(\llbracket{f}\rrbracket = \emptyset\)), we must return \(\bddFalse\) since \(\llbracket{\bddFalse}\rrbracket \equiv m_{\emptyset}\).
In line 2, since \(\llbracket{\mathcal{F}}\rrbracket \equiv \boolTrue\) (\(\llbracket{f}\rrbracket = \tAs\)), the largest set satisfying our claim is \(\mathcal{P}\) itself.
Otherwise, \(\llbracket{f}\rrbracket \neq \tAs\) and \(\llbracket{\mathcal{P}}\rrbracket \equiv m_{\pAs}\) (line 3). 
Let \(x_{j}\) be the variable rooted at \(\mathcal{F}\). 
Whenever \(x_{j}\) occurs in an assignment, we check whether it is fully covered by \(f_{x_{j}}\), and vice versa for \(\overline{x_{j}}\).
If \(x_{j}\) does not occur, it must be covered by both \(f_{x_{i}} \equiv \mathcal{F}_{x_{i}}\) and \(f_{\overline{x_{j}}} \equiv \mathcal{F}_{\overline{x_{j}}}\).
We now consider cases where \(\llbracket{\mathcal{P}}\rrbracket \not\equiv m_{\emptyset}\) and \(\llbracket{f}\rrbracket \neq \emptyset\).
Let \(x_{i}\) be the variable associated to the root label of \(\mathcal{P}\) and \(x_{j}\) the variable rooted at \(\mathcal{F}\).
When \(i > j\) (lines 12--13), we perform the same operations as when \(\llbracket{\mathcal{P}}\rrbracket \equiv m_{\pAs}\) and \(\llbracket{f}\rrbracket \neq \tAs\).
If \(i < j\) (lines 10--11), we need not consider both cofactors of \(\mathcal{F}\) when \(x_{i}\) does not occur in the assignment as \(x_{i}\) does not impact whether the assignment is covered by \(\mathcal{F}\) or not.
Whenever \(i = j\), we again consider the positive and negative cofactors of \(\mathcal{F}\) when having \(x_{i}\) and \(\overline{x_{i}}\) in the assignment, respectively.
Again, we must consider both cofactors of \(\mathcal{F}\) when dealing with non-occurrence of \(x_{i}\) in an assignment.\qed
\end{proof}

}

\section{Feature Causality}
One application for filtering assignment sets is in the context of explaining configurable systems.
Given a set of features and a set of total assignments over these features representing some effect, \emph{feature causality} identifies the features that are the reasons for exhibiting the effect~\cite{Dubslaff2022}.
The classical approach for computing these causes is through existentially filtering prime implicants.
We here demonstrate the use of our implicit filtering algorithms to compute a generalized notion of \emph{feature causes}~\cite{Dubslaff2024}.
\begin{definition}[\cite{Dubslaff2024}]
A \emph{feature cause} for effect configurations \(\mathsf{Effect} \subseteq \tAs\) w.r.t.\ non-effect configurations \(\mathsf{NEffect} \subseteq \tAs\) where \(\mathsf{Effect} \cap \mathsf{NEffect} = \varnothing\) is a partial assignment \(\mathbf{p} \in \pAs\) such that
\begin{enumerate}[leftmargin=*,ref={\textbf{(FPC\arabic*)}},label={\textbf{(FPC\arabic*)}}]
\item\label{fpc1} \(\llbracket{\mathbf{p}}\rrbracket \cap \mathsf{Effect} \neq \varnothing\), \(\llbracket{\mathbf{p}}\rrbracket \cap \mathsf{NEffect} = \varnothing\), and
\item\label{fpc2} \(\llbracket{\mathbf{p}[*/x]}\rrbracket \cap \mathsf{NEffect} \neq \varnothing\) for all \(x \in \mathsf{supp}(\mathbf{p})\).
\end{enumerate}
\end{definition}
\begin{figure}[t]
\centering    {
\noindent
\begin{tikzpicture}[
    nonEffect/.style={
        circle,
        draw=none,
        opacity=0.2,
        fill=orange,
    },
    effect/.style={
        circle,
        draw=none,
        opacity=0.2,
        fill=blue,
    },
    invalid/.style={
        circle,
        draw,
        fill=red!80!white,
        fill opacity=0.2,
    },
    valid/.style={
        circle,
        draw,
        thin,
        fill=green!80!white,
        fill opacity=0.4,
    },
]
\node[
    draw=none,
    fill=none,
    minimum width=\textwidth,
    minimum height=5em,
    outer sep=0pt,
    inner sep=0pt,
] (box) {};
\begin{scope}
    \clip (box.south west) -- (box.south east) -- (box.north east) -- (box.north west) -- cycle;

    \def\radiusEffect{45pt}
    \def\radiusNonEffect{\radiusEffect}

    \node[
        effect,
        minimum width=\dimexpr\radiusEffect*2\relax,
    ] (effectSpace) at ($(box.center) - (\dimexpr0.25\textwidth\relax,0)$) {};
    \node[
        scale=1,
        anchor=south
    ] (eText) at ($(effectSpace|-box.south)$) {\(\mathsf{Effect}\)};
    \node[
        nonEffect,
        minimum width=\dimexpr\radiusNonEffect*2\relax,
    ] (nonEffectSpace) at ($(box.center) + (\dimexpr0.25\textwidth\relax,0)$) {};
    \node[
        scale=1,
    ] (eText) at ($(nonEffectSpace|-eText)$) {\(\mathsf{NEffect}\)};

    \newlength{\vWd}\setlength{\vWd}{5em}
    \node[
        valid,
        minimum width=\vWd
    ] (v) at ($(effectSpace)!1!15:($(effectSpace)+(\radiusEffect,0)$)$) {};
    \node at (v) {\(\llbracket{\mathbf{p}}\rrbracket\)};
    \coordinate (uPt) at ($(v)!0.75!45:($(v)-(0.5\vWd,0)$)$);
    \fill (uPt) circle (1.5pt);
    \node[
        anchor=west
    ] at (uPt) {\(\mathbf{u}\)};

    \node[
        invalid,
        minimum width=4em
    ] (iv1) at ($(nonEffectSpace)!1!-15:($(nonEffectSpace)+(\radiusNonEffect,0)$)$) {};
    \node at (iv1) {\(\llbracket{\mathbf{r}}\rrbracket\)};
    \node[
        invalid,
        minimum width=4.2em,
    ] (iv2) at ($(box)+(15pt,0)$) {};
    \node at (iv2) {\(\llbracket{\mathbf{q}}\rrbracket\)};

    \node[
        anchor=north east,
        outer sep=0pt,
        scale=1.2,
    ] at (box.north east) { \(\Theta_{X}\) };
\end{scope}

\draw[] (box.south west) -- (box.south east) -- (box.north east) -- (box.north west) -- cycle;
\end{tikzpicture}
}
    \caption{Configuration space for feature causes.}\label{fig:fcause}
\end{figure}
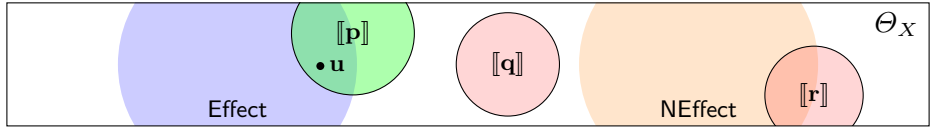
\Cref{fig:fcause} illustrates the configuration space for feature causes. 
Condition \ref{fpc1} is satisfied only by assignment \(\mathbf{p}\) as there exists a total assignment \(\mathbf{u}\) in \(\mathsf{Effect}\) in the cover of \(\mathbf{p}\) and no total assignment covered by \(\mathbf{p}\) is in \(\mathsf{NEffect}\).
Assignment \textbf{q} satisfies \(\llbracket{\mathbf{q}}\rrbracket \cap \mathsf{NEffect} = \varnothing\), but does not contain an instance in \(\mathsf{Effect}\), and \(\llbracket{\mathbf{r}}\rrbracket \cap \mathsf{NEffect} \neq \varnothing\).
Assignment \(\mathbf{p}\) would be a feature cause if removing any variable \(x \in \mathsf{supp}(\mathbf{p})\) from \(\mathbf{p}\) would result in covering an instance in \(\mathsf{NEffect}\).
Using our techniques, \Cref{algo:precause} computes these causes.

\newcommand{\pCause}{\ensuremath{\textsc{PCause}}}
\begin{algorithm}
\caption{$\pCause(\mathcal{F}_{\mathsf{Effect}},\mathcal{F}_{\mathsf{NEffect}})$}\label{algo:precause}
\textbf{Input: } BDDs \(\mathcal{F}_{\mathsf{Effect}}\), \(\mathcal{F}_{\mathsf{NEffect}}\) s.t.\ \(\llbracket{\mathcal{F}_{\mathsf{Effect}}}\rrbracket \equiv c_{\mathsf{Effect}}\), \(\llbracket{\mathcal{F}_{\mathsf{NEffect}}}\rrbracket \equiv c_{\mathsf{NEffect}}\), and \(\mathsf{Effect} \cap \mathsf{NEffect} = \varnothing\) \\
\textbf{Output: } meta-product BDD \(\llbracket{\mathcal{Q}}\rrbracket \equiv Q\) s.t.\ \(\mathbf{q}\) is a cause for \(\mathsf{Effect}\) w.r.t.\ \(\mathsf{NEffect}\) for all \(\mathbf{q} \in Q\).
\begin{algorithmic}[1]
\State{\(\mathcal{P} \leftarrow \Prime(\overline{\mathcal{F}_{\mathsf{NEffect}}})\)}
\State\Return{\(\eFilter(\mathcal{P}, \mathcal{F}_{\mathsf{Effect}})\)}
\end{algorithmic}
\end{algorithm}

Moreover, we may compute the set of most general covers~\cite{Dubslaff2024} using the implicit techniques we introduced in the last section.
Most-general causes filter those causes that are covered by another cause and thus would serve as shorter explanation for the effect.
For computing such covers, we present a generalization based on $\forall$-filtering in \Cref{algo:mpGeneralize}, which removes primes covered by other primes.
This is done by repeatedly taking a candidate assignment and extracting a more general assignment w.r.t.\ space \(f\).
Candidate (partial) assignments are represented by meta-product BDDs; the function \(\mathsf{OneSAT}(\mathcal{Q})\) obtains a satisfying assignment of the input meta-product BDD \(\mathcal{Q}\), which in itself is a meta-product BDD representing a singleton assignment from \(Q\), i.e.\ \(\mathsf{OneSAT}(\mathcal{Q}) = \mathcal{B}_{\mathbf{p}}\) such that \(\llbracket{\mathcal{Q}}\rrbracket \equiv m_{Q}\), \(\llbracket{\mathcal{B}_{\mathbf{p}}}\rrbracket \equiv m_{\{\mathbf{p}\}}\), and \(\mathbf{p} \in Q\).
Applying \(\mathsf{cover}\) on a meta-product BDD is the cover of the represented products, i.e.\ \(\mathsf{cover}(\mathcal{P}) = \mathcal{F}\) such that \(\llbracket{\mathcal{F}}\rrbracket \equiv f\), \(\llbracket{\mathcal{P}}\rrbracket \equiv m_{P}\), and \(\llbracket{f}\rrbracket = \llbracket{P}\rrbracket\); this operation is implemented by existentially removing all occurrence variables from the meta-product and re-assigning the sign variables to the variable domain of \(\mathcal{F}\).
The resulting output of \Cref{algo:mpGeneralize} is an equivalence class of causes that are cannot be further generalized. 
To obtain all general causes, we run \Cref{algo:mostGeneral}.
The idea is to pick a candidate prime (line 2), generalize them to obtain the most general equivalence class covering this candidate prime, and add these general causes to the set \(\mathcal{Q}\) (line 3).
Afterward, we remove the remaining products that cover the same space in \(f\) as our computed equivalence class cover and repeat the procedure until we have no more candidates.

\begin{example}\label{ex:email}
\sloppy
Consider the email example from~\cite{Dubslaff2024}.
The system has features \(F = \{m,s,e,c,a,r\}\) which formalizes the base e\underline{m}ail function, and optional features for \underline{s}igning and \underline{e}ncrypting.
If one encrypts, one must choose exactly one encryption method: \underline{C}aesar, \underline{A}ES, or \underline{R}SA.
\Cref{fig:emailFeatures} shows the \emph{feature diagram}~\cite{Kang1990} of the valid configurations: the feature at the root of the diagram must be included, and nodes with \(\circ\) are optional features.
The arc over edges represent exclusive disjunctions over the connected children nodes.
Here, signing and encryption are optionally enabled features; however, if one chooses to encrypt, one must select exactly one encryption algorithm.
Let us call the set of valid configurations \(\mathsf{Valid} = \{m, mec, mea, mer, ms, msec, msea, mser\}\).
Here, each valid configuration is a total assignment over \(F\) where negatively set variables are omitted, e.g.\ \(m\) represents the total assignment \(m\overline{secar}\).
We assume complete knowledge of the system, i.e.\ we are aware of all valid configurations and all valid configurations for which the effect may be observed, and thus \(\mathsf{Valid} = \mathsf{Effect} \cup \mathsf{NEffect}\).

Now assume that we are interested in the property ``long decipher time'', and let the property hold whenever AES or RSA is enabled, i.e.\ \(\mathsf{Effect} = \{mea,mer,msea,mser\}\).
Again, each effect is a total assignment over \(F\) with negatively set variables omitted.
Applying prime implicants over \(\overline{\mathsf{NEffect}} = \overline{\mathsf{Valid}\setminus\mathsf{Effect}}\), yields the following (partial) assignments: \(\overline{m},\ \overline{e}c,\ e\overline{c},\ a,\ r\).
Filtering these primes to relevant causes yields the following feature causes: \(e\overline{c},\ a,\ r\).
In other words, enabling encryption and not choosing the Caesar method will make the effect observable. 
Moreover, having AES or RSA enabled also makes the effect observable.

Given a set of feature causes, we may provide the user a subset of causes which sufficiently cover all effects. 
Observe that \(e\overline{c}\) covers \(a\) and \(r\) w.r.t.\ \(\mathsf{Valid}\), i.e.\ \((\llbracket{a}\rrbracket \cap \mathsf{Valid}) \subset (\llbracket{e\overline{c}}\rrbracket \cap \mathsf{Valid})\) and \((\llbracket{r}\rrbracket \cap \mathsf{Valid}) \subset (\llbracket{e\overline{c}}\rrbracket \cap \mathsf{Valid})\).
Thus, \(\{e\overline{c}\}\) is the set of most general causes.
\end{example}

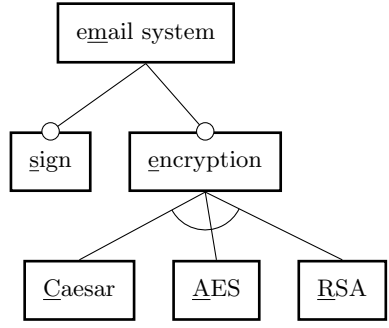
\begin{figure}[t]
\centering\newlength{\unit}
\setlength{\unit}{3em}
\begin{tikzpicture}[
		feature/.style={rectangle, line width=1pt,draw,inner sep=.75em}
	]
	\node[feature] (emailSystem) {%
		e\underline{m}ail system
	};

	\begin{scope}[anchor=north, shift={($(emailSystem.south) + (0,-.8\unit)$)}]
		\matrix[column sep=.5\unit]{
			\node[feature] (sign) {%
				\underline{s}ign
			}; &
			\node[feature] (encrypt) {%
				\underline{e}ncryption
			}; \\
		};
	\end{scope}

	\begin{scope}[anchor=north, shift={($(encrypt.south) + (0,-.8\unit)$)}]
		\matrix[column sep = .5\unit]{
			\node[feature] (caesar) {%
				\underline{C}aesar
			};
			 &
			\node[feature] (aes) {%
				\underline{A}ES
			};
			 &
			\node[feature] (rsa) {%
				\underline{R}SA
			}; \\
		};
	\end{scope}

	\draw[] (emailSystem.south) -- (sign.north);
	\node[circle,draw,fill=white,inner sep=.25em,anchor=center] at (sign.north) {};
	\draw[] (emailSystem.south) -- (encrypt.north);
	\node[circle,draw,fill=white,inner sep=.25em,anchor=center] at (encrypt.north) {};

	\draw[] (encrypt.south) -- (caesar.north);
	\draw[] (encrypt.south) -- (aes.north);
	\draw[] (encrypt.south) -- (rsa.north);

	\coordinate (a) at (caesar.north);
	\coordinate (b) at (encrypt.south);
	\coordinate (c) at (rsa.north);
	\draw pic[draw] {angle=a--b--c};
\end{tikzpicture}
    \caption{Feature diagram of the email example~\cite{Dubslaff2024}}\label{fig:emailFeatures}
\end{figure}

\def\gen{\textsc{GeneralCover}}
\begin{algorithm}
    \algrenewcommand\algorithmicrepeat{\textbf{do}}
    \algrenewcommand\algorithmicuntil{\textbf{while}}

    \floatname{algorithm}{Algorithm}
    \algrenewcommand\algorithmicrequire{\textbf{Input:}}
    \algrenewcommand\algorithmicensure{\textbf{Output:}}
    \caption{$\gen(\mathcal{P},\mathbf{p},\mathcal{F})$}\label{algo:mpGeneralize}
    \begin{algorithmic}[1]
        \Require~{meta-product BDD \(\mathcal{P}\) s.t.\ \(\llbracket{\mathcal{P}}\rrbracket \equiv m_{P}\), BDD \(\mathcal{B}_{\mathbf{p}}\) where \(\llbracket{\mathcal{B}_{\mathbf{p}}}\rrbracket \equiv m_{\{\mathbf{p}\}}\) for \(\mathbf{p} \in P\), and BDD \(\mathcal{F}\).}
        \Ensure~{meta-product BDD $\mathcal{Q}$.}
        \State{$\mathcal{Q} \leftarrow \mathcal{P}$; $\mathcal{B}_{\mathbf{q}} \leftarrow \mathcal{B}_{\mathbf{p}}$}
        \Repeat{} 
        \State{$\mathcal{Q} \leftarrow \mathcal{Q} \land \lnot \mathcal{B}_{\mathbf{q}}$}
        \State{$\mathcal{Q} \leftarrow \aFilter(\mathcal{Q},(\mathsf{cover}(\mathcal{B}_{\mathbf{q}}) \land \mathcal{F}))$}
        \Until{\(\mathcal{B}_{\mathbf{q}} \leftarrow \mathsf{OneSAT}(\mathcal{Q})\)}
        \State{\Return{$\aFilter(\mathcal{Q},(\mathsf{cover}(\mathcal{B}_{\mathbf{q}}) \land \mathcal{F}))$}}
    \end{algorithmic}
\end{algorithm}

\def\mostGen{\textsc{AllGeneralCovers}}
\begin{algorithm}
    \floatname{algorithm}{Algorithm}
    \algrenewcommand\algorithmicrequire{\textbf{Input: }}
    \algrenewcommand\algorithmicensure{\textbf{Output: }}
    \caption{$\mostGen(\mathcal{P}, \mathcal{F})$}\label{algo:mostGeneral}
    \begin{algorithmic}[1]
        \Require~{meta-product BDD $\mathcal{P}$ and BDD \(\mathcal{F}\)}
        \Ensure~{meta-product BDD $\mathcal{Q}$}
        \State{$\mathcal{Q} \leftarrow \bddFalse$; $\mathcal{P}' \leftarrow \mathcal{P}$}
        \While{$\mathcal{B}_{\mathbf{p}} \leftarrow \mathsf{OneSAT(\mathcal{P}')}$}
            \State{$\mathcal{Q} \leftarrow \mathcal{Q} \lor \gen(\mathcal{P}',\mathcal{B}_{\mathbf{p}},\mathcal{F})$}
            \State{$\mathcal{P}' \leftarrow \mathcal{P}' \land \lnot {\cFilter(\mathcal{Q},(\lnot{\mathcal{F}} \lor \mathsf{cover}(\mathcal{B}_{\mathbf{p}})))}$}
        \EndWhile{}
        \State{\Return{$\mathcal{Q}$}}
    \end{algorithmic}
\end{algorithm}

\section{Toolkit}
We have implemented our tools in C++ using the BuDDy BDD library~\cite{LindNielsen1999}.
The repository may be found using the following link:
\begin{center}
\texttt{https://gitlab.tue.nl/201810931/filtered-primes}
\end{center}
In this section, we provide a short overview of our tools.
\begin{center}
    \renewcommand\arraystretch{1.5}
    \addtolength{\tabcolsep}{0.5em}
    \begin{NiceTabularX}{\textwidth}{@{}>{\ttfamily}c X@{}}
    \toprule
    \normalfont{Program} & Description\\
    \midrule
    \rowcolors{gray!10}{}[respect-blocks]
    \Block[]{1-1}{formula2bdd} & Compile propositional logic formulas into BDDs.\\
    \Block[]{1-1}{bddOp} & Applies BDD operations to provided BDD file(s).\\
    \Block[]{1-1}{primes} & Coudert and Madre's implicit prime implicants.\\
    \Block[]{1-1}{pfilter} & Filters a meta-product BDD.\\
    \Block[]{1-1}{i2e} & Implicit to explicit assignments.\\
    \midrule
    \Block[m]{1-1}{fcause} & All-in-one tool to compute filtered prime implicants for feature (pre)causes. Combines implicit prime computation and filtering.\\
    \Block[]{1-1}{generalize} & Finds the set of most general causes.\\
    \Block[m]{1-1}{mgCandidate} & Given the set of most general causes, finds a candidate that sufficiently covers relevant instances.\\
    \bottomrule
    \end{NiceTabularX}
\end{center}

\subsubsection{Example Use Case.}
Using our tools, we provide an example of computing the feature causes of the email system taken from \Cref{ex:email}.
This example is also provided within the source code repository of our tools.

\begin{example}
Assume we are given the following files:
\begin{center}
\setlength{\tabcolsep}{10pt}
\begin{NiceTabularX}{\textwidth}{@{}>{\ttfamily}c X@{}}
\toprule
\normalfont{File Name} & Description \\
\midrule
\rowcolors{gray!10}{}[respect-blocks]
\Block[]{1-1}{email.fs} & A list of feature names, each on their own line.\\
\Block[]{1-1}{on.dnf} & A DNF formula representing the ON-set as \(\mathsf{On}\) we witness, where each disjunct/product is on its own line.\\
\Block[]{1-1}{valid.dnf} & A DNF formula representing the configuration space as the care set \(\mathsf{Valid}\).\\
\bottomrule
\end{NiceTabularX}
\end{center}
First, we obtain the BDDs from the formula files.
Running the following command gives the file \texttt{on.bdd} and \texttt{valid.bdd}.
\begin{lstlisting}
    formula2bdd ./email.fs ./on.dnf ./on.bdd
    formula2bdd ./email.fs ./valid.dnf ./valid.bdd
\end{lstlisting}

Since we are computing feature causes, we must first obtain the relevant BDDs for prime computation.
As \(\mathsf{On}\) may contain invalid configurations, we restrict the \(\mathsf{On}\) BDD to the set of \(\mathsf{Valid}\) configurations to obtain the \(\mathsf{Effect}\) BDD, i.e.\ \(\mathsf{Effect} = \mathsf{Valid} \cap \mathsf{On}\).
For feature causes, we must first find the primes of \(\overline{\mathsf{Valid}} \cup \mathsf{Effect}\).
Thus, using \texttt{bddOp} we first obtain \(\overline{\mathsf{Valid}}\), then our prime implicant search space \(\overline{\mathsf{Valid}} \cup \mathsf{Effect}\).
After which, we compute the prime implicants.
\begin{lstlisting}
    bddOp --and ./valid.bdd ./on.bdd ./effect.bdd
    bddOp --not ./valid.bdd ./nValid.bdd
    bddOp --or ./nValid.bdd ./on.bdd ./nV_or_e.bdd
    primes ./email.fs ./nV_or_e.bdd ./candidateCauses.bdd
\end{lstlisting}

To obtain the feature causes, we use our \(\eFilter\) algorithm on the candidate causes over the \(\mathsf{Effect}=\mathsf{Valid} \cap \mathsf{On}\) configuration space.
\begin{lstlisting}
    pfilter --exists ./candidateCauses.bdd ./effect.bdd ./fCauses.bdd
\end{lstlisting}

As the meta-product BDD implicitly represents the causes, we may obtain a list of causes using the \texttt{i2e} tool:
\begin{lstlisting}
    i2e ./email.fs ./fCauses.bdd
\end{lstlisting}

We obtain the following output after the procedure, comprising all the feature causes and explicitly listing them as partial assignments:\\[2pt]
\noindent
\tikz{\node[
text width=\textwidth,
align=left,
fill=gray!20,
font=\ttfamily,
outer sep=0pt
]{
rsa\\
encrypt {\textasciitilde}caesar\\
aes
};
}
In \Cref{ex:email}, we have that (\(r,\ a,\ e\overline{c}\)) are the feature causes of the email system.
Observe that we properly obtain all relevant feature causes using our tool.

Going further, one may compute the largest set of causes that produce a most general cover:
\begin{lstlisting}
    generalize ./email.fs ./valid.bdd ./on.bdd ./fCauses.bdd ./genCauses.bdd
    i2e ./email.fs ./genCauses.bdd
\end{lstlisting}
\noindent
This yields the following output:\\[2pt]
\noindent
\tikz{\node[
text width=\textwidth,
align=left,
fill=gray!20,
font=\ttfamily,
outer sep=0pt
]{
encrypt {\textasciitilde}caesar
};
}
We see that our tool obtains a set of most general causes.
In this example we obtain the singleton set consisting of one cause: \(e\overline{c}\).
The other causes do not appear as they are subsumed by \(e\overline{c}\), w.r.t.\ the valid configurations.

\subsubsection{Feature Causality Convenience.} Alternatively, for directly computing feature causes without manually executing the above pipeline, one may use the dedicated \texttt{fcause} tool to obtain the feature causes:
\begin{lstlisting}
    fcause ./email.fs ./valid.bdd ./on.bdd ./fCauses.bdd
    i2e ./email.fs ./fCauses.bdd
\end{lstlisting}
These commands yield the same output as the pipeline illustrated above.

\end{example}

\section{Conclusion}\label{sec:conclusion}
We have presented algorithms over the structure of meta-product BDDs to filter sets of assignments implicitly based on various criteria, exemplified by existential, universal, and suset-filtering.
Exemplified by the use case of feature causality, we show how our techniques enable a completely symbolic computation pipeline for deriving feature causes and the set of most general causes.
We offer a software toolkit of our methods, allowing for BDD construction, implicit computation of prime implicants, and filtering symbolic assignment sets.

Future work includes an empirical evaluation of our algorithms by comparing our fully symbolic pipeline with the partially symbolic approach for feature causes~\cite{Dubslaff2024} and integrate performant BDD-based techniques to compile care sets and ON-sets, e.g.\ required for symbolic feature model and effect set representations~\cite{Husung2024,DubHusKaf26}.

\paragraph{Acknowledgements.}
This work was partially supported by the DFG (TRR 248, see \url{https://perspicuous-computing.science}, project ID 389792660) and the NWO through Veni grant VI.Veni.222.431.

\bibliographystyle{splncs04}
\bibliography{./main.bib}

\end{document}